\documentclass[11pt,a4paper]{article}
\usepackage[a4paper,margin=1in]{geometry}
\usepackage[T1]{fontenc}
\usepackage[utf8]{inputenc}
\usepackage{amsmath,amssymb,amsthm,mathtools,bm}
\usepackage{mathrsfs}
\usepackage{booktabs,tabularx,array}
\usepackage{enumitem}
\usepackage{microtype}
\usepackage{xcolor}
\usepackage{float}
\usepackage{tikz}
\usepackage{pgfplots}
\pgfplotsset{compat=1.18}
\usepackage[colorlinks=true,linkcolor=blue!60!black,citecolor=blue!60!black,urlcolor=blue!60!black]{hyperref}
\usepackage[nameinlink,noabbrev]{cleveref}
\setlist[itemize]{topsep=3pt,itemsep=2pt,leftmargin=2em}
\setlist[enumerate]{topsep=3pt,itemsep=2pt,leftmargin=2em}
\allowdisplaybreaks
\theoremstyle{plain}
\newtheorem{theorem}{Theorem}[section]
\newtheorem{proposition}[theorem]{Proposition}

\newtheorem{corollary}[theorem]{Corollary}
\theoremstyle{definition}
\newtheorem{definition}[theorem]{Definition}

\theoremstyle{remark}
\newtheorem{remark}[theorem]{Remark}

\newcommand{\R}{\mathbb{R}}
\newcommand{\C}{\mathbb{C}}
\newcommand{\E}{\mathbb{E}}

\newcommand{\Var}{\operatorname{Var}}
\newcommand{\Tr}{\operatorname{Tr}}
\newcommand{\rank}{\operatorname{rank}}
\newcommand{\diag}{\operatorname{diag}}
\newcommand{\Comm}{\operatorname{Comm}}
\newcommand{\Herm}{\operatorname{Herm}}
\newcommand{\End}{\operatorname{End}}
\newcommand{\Span}{\operatorname{span}}
\newcommand{\MMD}{\operatorname{MMD}}
\newcommand{\FIM}{\mathbf F}

\newcommand{\calH}{\mathcal H}
\newcommand{\calO}{\mathcal O}
\newcommand{\calS}{\mathcal S}
\newcommand{\calX}{\mathcal X}
\newcommand{\calY}{\mathcal Y}
\newcommand{\calZ}{\mathcal Z}
\newcommand{\calF}{\mathcal F}
\newcommand{\calL}{\mathcal L}
\newcommand{\calA}{\mathcal A}
\newcommand{\calT}{\mathcal T}
\newcommand{\push}{\mathbin{\#}}
\newcommand{\norm}[1]{\left\lVert #1\right\rVert}
\newcommand{\abs}[1]{\left\lvert #1\right\rvert}
\newcommand{\ip}[2]{\left\langle #1,#2\right\rangle}
\newcommand{\ket}[1]{\left|#1\right\rangle}
\newcommand{\bra}[1]{\left\langle#1\right|}
\newcommand{\braket}[2]{\left\langle#1\middle|#2\right\rangle}

\title{Symmetry Discovery in Quantum Learning:\\
  Observable-Level and Task-Level Inference from Finite Measurements}
\author{Zeyu Chen}
\date{}

\begin{document}
\maketitle

\begin{abstract}
Symmetry reduces the capacity of a quantum learning model, but the imposed group must match both the measured information and the label transformation. We establish a finite-measurement theory for inferring this group from candidate transformations. The central structural result identifies observable-invisible transformations with the stabilizer of a projected state whenever the probe span is invariant. It turns recovered generators into a valid subgroup and identifies the continuous invisible space with its Lie algebra. For finite dictionaries, an unbiased shadow statistic distinguishes zero from positive squared expectation discrepancies with an inverse-gap measurement rate, improving the inverse-square-gap rate of uniform discrepancy estimation. A commuting qubit lower bound proves the gap dependence optimal at fixed snapshot scale, and simultaneous intervals support data-dependent tolerances. Task validation then tests either the joint distribution through a characteristic kernel or its encoded mean through a classical--quantum discrepancy. An exact group-average identity relates the latter to joint-state asymmetry and specifies its conversion to binary task breaking mass. Projection bias quantifies the cost of excessive symmetry, while an $\ell_1$ readout bound quantifies the capacity gained by relaxing it. At an invariant pure-state backbone, retained and nontrivial breaking sectors are Fisher-orthogonal. Ising-chain calculations connect finite-shot recovery, label-dependent symmetry, and physical sector drift. These results determine which symmetry the measurements support and provide the statistical and geometric basis for a subsequent release decision.

\end{abstract}

\section{Introduction}

A symmetry constraint determines which distinctions a quantum learning model can represent. Equivariant circuits use a prescribed group to organize their trainable operators and can improve statistical efficiency and trainability under suitable architectural assumptions \cite{CongChoiLukin2019QCNN,Meyer2023,Nguyen2024}. This structure matters in quantum models, where the reachable operator algebra also governs barren-plateau mechanisms \cite{McClean2018Barren,Larocca2022DLA,Ragone2024}. Its value nevertheless depends on the task. A spin-flip-symmetric ensemble can carry an odd magnetization label: an equivariant predictor should reverse that label, whereas an invariant predictor suppresses the distinction to be learned.

Inferring a suitable constraint requires separate evidence about the state, the measured observables, and the input--label law. Local probes may miss a state change, and an invariant ensemble need not support invariant labels. State-preserving parameter redundancy introduces a further, geometric notion of symmetry. The central question is therefore which physical transformations finite measurements can certify, when recovered generators define a group, and how that group should act on the labels.

Our answer starts with the information retained by the measurement interface. If the probe span is invariant under the candidate action, its Hilbert--Schmidt projection of the state contains exactly the information needed for observable-level symmetry. The invisible transformations form the stabilizer of this projected state; their infinitesimal null space is its Lie algebra. Exact recovered generators consequently generate an invisible subgroup, while approximate generator discrepancies control finite words. The invariant-span condition also identifies the obstruction for restricted probes: a bracket-closed tangent null space can still permit second-order changes of the measured expectations.

The statistical consequence is a sharp gap dependence for finite-dictionary recovery. With bounded classical-shadow snapshots, an unbiased quadratic statistic distinguishes zero from positive squared discrepancies using $M=O(R_{\calO}^2\gamma_{\calO}^{-1}\log(m|\calS_0|/\delta))$ measurements. Here $m$ is the probe count, $\calS_0$ the fixed candidate dictionary, $R_{\calO}$ a discrepancy-snapshot bound, and $\gamma_{\calO}$ the positive squared-discrepancy gap. The improvement over uniform additive estimation follows from the smaller fluctuation scale at a zero mean. A commuting qubit construction proves the inverse-gap dependence necessary, even with collective measurements. Simultaneous intervals extend the same analysis to tolerances chosen from the data, and a covariance estimator gives stable recovery of continuous invisible directions.

Task validation determines how the inferred transformation enters the predictor. A characteristic-kernel discrepancy identifies invariance of the joint input--label distribution. For quantum inputs, a classical--quantum mean-state discrepancy admits paired-copy and reusable-shadow estimators, and its group average equals the asymmetry mass of the joint representation. The binary specialization separates class-average asymmetry from the symmetry sector carrying the label. These identities explain why the same spin flip can be retained with a sign-transforming label and rejected with an invariant label. Projection bias and an $\ell_1$ readout bound quantify the corresponding approximation and capacity costs; representation theory supplies Fisher orthogonality between a symmetric pure-state backbone and nontrivial release sectors.

The resulting inference layer supplies a retained physical group, a supported label action, and confidence bounds at the chosen resolution. Statistical release theory uses these objects to decide whether and how far to activate breaking directions \cite{CompanionRelease}. Quantum natural gradient (QNG) on quotient spaces handles the subsequent optimization once state-preserving parameter redundancy has been independently identified \cite{ChenQNGQuotient2026}. The Ising example connects these roles through explicit recovery curves, task discrepancies, Fisher blocks, and sector drift.

\subsection*{Related work}

Quantum convolutional and equivariant architectures impose symmetry as a prior \cite{CongChoiLukin2019QCNN,Meyer2023,Nguyen2024}; dynamical Lie-algebra methods characterize associated trainability mechanisms \cite{Larocca2022DLA,Ragone2024}. Variational quantum--classical algorithms instead search for unitaries that stabilize an unknown state \cite{Lu2024}. Classical methods infer transformation orbits or Lie generators from data and learned representations \cite{Krippendorf2020,Yang2023LieGAN,Moskalev2022LieGG}. Our recovery results concern a specified quantum measurement interface and a fixed finite dictionary or candidate Lie algebra, with explicit sampling margins and a separate test of the label action.

Quantum symmetry tests examine a prescribed state or channel action and its computational complexity \cite{Laborde2021Testing,Laborde2023Complexity}. Hilbert--Schmidt asymmetry estimators use two-copy overlaps \cite{Bandyopadhyay2023}, while coherent projector constructions resolve symmetric and other representation subspaces \cite{LaBordeRethinasamyWilde2024}. Full subgroup inference is also studied through the state hidden subgroup problem \cite{Bouland2025StateHSP}; its abelian algorithm uses representation measurements and achieves an inverse-overlap-gap copy bound \cite{Hinsche2026StateHSP}. These access models test full-state symmetry. Our projected-state characterization specifies the group visible to a restricted probe span and the additional assumptions needed to identify it with a physical stabilizer.

Classical shadows provide simultaneous observable estimates \cite{Huang2020ClassicalShadows}. Calibrated noise-aware variants modify their snapshot and variance bounds \cite{Chen2021RobustShadows,Koh2022RobustShadows}, and symmetry-adapted shadows reduce measurement cost when suitable symmetry information is already available \cite{SauvageLarocca2024}. Shadow-based singular-value methods learn local conservation laws of unknown dynamics \cite{Zhan2024Conservation}; the covariance used here tests infinitesimal candidate actions on a specified state. Equivariant-map learning between supplied representations addresses the subsequent construction of an embedding \cite{Bradshaw2025Maps}.

Relaxed equivariance in classical architectures uses residual pathways or learned mixing to accommodate imperfect symmetry \cite{Finzi2021RPP,Wang2022RelaxedEquivariance}. The statistical release theory of \cite{CompanionRelease} evaluates task gain, measurement uncertainty, and model selection after such directions are proposed. Our role is to infer the retained constraint and characterize its release geometry. The kernel tests, concentration arguments, and readout-complexity estimates draw on established statistical tools; the projected-state subgroup criterion, gap-sensitive recovery, and input--label bridge connect those tools to the distinct symmetry statements required by a quantum learning model.

\section{Setting and separation of symmetry levels}
\label{sec:setting}

Let $G$ be a finite or compact group, represented continuously by unitaries $g\mapsto U_g$ on a finite-dimensional Hilbert space $\calH$. All subgroup averages use normalized Haar measure on closed subgroups. The Hilbert--Schmidt inner product is $\ip{A}{B}_{\mathrm{HS}}=\Tr(A^\dagger B)$ with unnormalized trace; $\|\cdot\|_{\mathrm{HS}}$ and $\|\cdot\|_{\mathrm{op}}$ denote the Hilbert--Schmidt and operator norms. All confidence levels satisfy $0<\delta<1$.

For distributional statements, inputs and labels lie in Polish spaces with their Borel probability laws, and $G$ acts continuously on both spaces. For $z=(x,y)\in\calZ:=\calX\times\calY$, write
\begin{equation}
T_g(x,y):=(g\cdot x,\rho_{\calY}(g)y),
\label{eq:Tg}
\end{equation}
where $\rho_{\calY}$ is an orthogonal or unitary output representation.  For a distribution $P$ on $\calZ$, the pushforward under $T_g$ is written $T_g\push P$.

A \emph{state-level symmetry} of a density operator $\rho$ is an element $g$ such that
\begin{equation}
U_g\rho U_g^\dagger=\rho.
\end{equation}
An \emph{observable-level symmetry relative to $\calO$} is an element that leaves the expectations of a chosen observable family $\calO$ unchanged.  A \emph{task-level symmetry} is an element $g$ such that $T_g\push P=P$ for the joint distribution of supervised examples.  Finally, a \emph{parameter redundancy} is a group action on the parameter manifold that changes coordinates without changing the represented physical state or task function.  The last notion is handled by quotient-space QNG, while the first three are the discovery objects studied here.

Finite-dictionary recovery uses a nonempty finite dictionary $\calS_0\subseteq G$ and probe or feature maps fixed independently of its estimation samples. The candidate Lie algebra in continuous recovery obeys the same independence condition. These objects may be proposed using independent pilot data. A dictionary selected from a larger fixed family can reuse samples when the confidence event covers that entire family. The dictionary-relative target is
\begin{equation}
\calS_0\cap H_\star,
\end{equation}
where $H_\star$ is the relevant stabilizer at the level being tested.  At the state and task levels $H_\star$ is a subgroup, so $\calS_0\cap H_\star$ is closed under the group operation as far as the dictionary reaches.  At the observable level it need not be: $\{g:D_{\calO}(g;\rho)=0\}$ can contain $g$ and $g'$ while excluding $gg'$, because invisibility to a fixed probe family is not preserved by composition.  The invariant-span condition in \cref{thm:probe_stabilizer} guarantees closure and integration. Outside that condition the generated finite action needs further validation; algebraic or bracket closure alone does not establish its invisibility.  Under a subgroup guarantee, the model can impose $\overline{\langle\widehat{\calS}\rangle}$, the closure of the generated subgroup; continuity preserves exact invariance under closure. This subgroup need not exhaust the stabilizer. Approximate acceptance controls only the tested tolerance and, through \cref{thm:probe_stabilizer}, words of bounded length; it does not establish exact invariance of the generated group.

\section{Observable-level symmetry discovery from finite quantum measurements}
\label{sec:observable_discovery}

The input to quantum-native discovery is a quantum state, or a process snapshot whose relevant information is summarized by expectation values, together with a finite family of observables; the output is the subset of candidate generators invisible to those observables.

\subsection{Operator-algebraic discrepancy}

Fix a family of Hermitian observables
\begin{equation}
\calO=\{O_1,\ldots,O_m\}.
\end{equation}
Typical probes are Pauli strings, local projectors, order parameters, or sector projectors. For qubits, $X,Y,Z$ are the Pauli matrices and $X_i,Y_i,Z_i$ their site embeddings; write $M_x=\sum_iX_i$, $M_y=\sum_iY_i$, and $M_z=\sum_iZ_i$.  For a candidate generator $g\in\calS_0$, define
\begin{equation}
\widetilde O_j^{(g)}:=U_g^\dagger O_jU_g,
\qquad
\Delta_j^{(g)}:=O_j-\widetilde O_j^{(g)}.
\end{equation}

\begin{definition}[Observable symmetry discrepancy]
For a density operator $\rho$, define
\begin{equation}
D_{\calO}(g;\rho):=\frac1m\sum_{j=1}^m\left(\Tr(\Delta_j^{(g)}\rho)\right)^2.
\label{eq:Dobs}
\end{equation}
The element $g$ is called $\calO$-invisible on $\rho$ if $D_{\calO}(g;\rho)=0$.
\end{definition}

The statistic is deliberately observable-relative.  If $\calO$ is nonseparating, then $D_{\calO}(g;\rho)=0$ certifies invisibility to the measured linear span; full invariance of $\rho$ requires a separating probe family.  A probe family that separates the candidate orbit $\{U_g\rho U_g^\dagger:g\in G\}$ suffices for a full-state conclusion on that orbit. Informational completeness is a sufficient condition for arbitrary states, with its corresponding measurement cost.

\begin{theorem}[When invisible generators define a subgroup]
\label{thm:probe_stabilizer}
Let $\mathcal V=\Span_{\R}\calO$ be invariant under conjugation by the candidate compact group $G$, and let $\Pi_{\mathcal V}$ be its Hilbert--Schmidt orthogonal projector. Set $\sigma=\Pi_{\mathcal V}\rho$, which need not be a density operator. Then
\begin{equation}
\{g\in G:D_{\calO}(g;\rho)=0\}
=\{g\in G:U_g\sigma U_g^\dagger=\sigma\}=:H_{\mathcal V,\rho}
\label{eq:probe_subgroup}
\end{equation}
is a closed subgroup. For a compact Lie candidate group with represented Hermitian algebra $\mathfrak g$, the infinitesimal invisible directions are exactly $\{L\in\mathfrak g:[L,\sigma]=0\}$, the Lie algebra of this subgroup. If the probes are orthogonal with common squared Hilbert--Schmidt norm $c$, then
\begin{equation}
D_{\calO}(g;\rho)=\frac{c}{m}\|\sigma-U_g\sigma U_g^\dagger\|_{\mathrm{HS}}^2,
\qquad
\sqrt{D_{\calO}(g_1\cdots g_\ell;\rho)}
\le\sum_{j=1}^{\ell}\sqrt{D_{\calO}(g_j;\rho)}.
\label{eq:probe_word_bound}
\end{equation}
Thus exact invisible generators generate an invisible subgroup, and generators with discrepancy at most $\tau$ have discrepancy at most $\ell^2\tau$ on words of length $\ell$.
\end{theorem}

\begin{proof}
Conjugation is orthogonal on Hermitian operators. Invariance of $\mathcal V$ under the group and its inverses therefore makes $\mathcal V^\perp$ invariant as well, so $\Pi_{\mathcal V}$ commutes with conjugation. Equality of all probe expectations is equivalent to $\Pi_{\mathcal V}(\rho-U_g\rho U_g^\dagger)=0$, which is precisely $\sigma-U_g\sigma U_g^\dagger=0$. A stabilizer of a continuous group action is a closed subgroup. Differentiation gives $[L,\sigma]=0$; Jacobi's identity closes this space under the Hermitian bracket $i[L,L']$, and exponentiation fixes $\sigma$. Parseval's identity gives the first formula in \eqref{eq:probe_word_bound}. The identity $\sigma-\pi(gh)\sigma=(\sigma-\pi(g)\sigma)+\pi(g)(\sigma-\pi(h)\sigma)$, with $\pi(g)X=U_gXU_g^\dagger$, gives the second by the triangle inequality and iteration.
\end{proof}

The span condition is an operator calculation that can be checked before measuring the state. All Pauli strings of weight at most $r$ span an invariant space under on-site rotations and site permutations. For unrestricted entangling candidate unitaries, their orbit span may be larger and its measurement cost must be included. Bracket closure of a restricted-probe null space alone is insufficient: for $\rho=(I+Z)/2$, probe $Z$ and candidate algebra $\Span_{\R}\{X\}$, the first derivative vanishes and the null space is bracket closed, yet the probe expectation along $e^{-itX}$ is $\cos(2t)$. The theorem supplies the condition that rules out this second-order escape.

\subsection{Classical-shadow estimator and recovery guarantee}

Suppose each $\Delta_j^{(g)}$ can be estimated from a classical-shadow snapshot $\widehat\rho_\ell$ by
\begin{equation}
X_{\ell,j}^{(g)}:=\Tr(\Delta_j^{(g)}\widehat\rho_\ell),
\qquad
\E X_{\ell,j}^{(g)}=\Tr(\Delta_j^{(g)}\rho).
\end{equation}
For Pauli or Clifford shadows, the measurement cost is governed by the observable family \cite{Huang2020ClassicalShadows}. Assume $M\ge2$ independent snapshots with a known almost-sure bound $|X_{\ell,j}^{(g)}|\le R_{\calO}$ for every tested pair $(g,j)$. This bounded model includes finite-dimensional Pauli and Clifford snapshots and keeps the mean, variance, and tail scales explicit.

To avoid the positive bias of squaring a noisy mean, use the second-order $U$-statistic \cite{Hoeffding1948}
\begin{equation}
\widehat q_j^{(g)}:=\frac{1}{M(M-1)}\sum_{\ell\ne r}X_{\ell,j}^{(g)}X_{r,j}^{(g)},
\qquad
\widehat D_{\calO}(g):=\frac1m\sum_{j=1}^m\widehat q_j^{(g)}.
\label{eq:shadow_u}
\end{equation}
Then $\E\widehat q_j^{(g)}=(\Tr(\Delta_j^{(g)}\rho))^2$, so $\widehat D_{\calO}(g)$ is unbiased for $D_{\calO}(g;\rho)$.

\begin{theorem}[Observable-level dictionary recovery]
\label{thm:shadow_recovery}
Let $\calS_\star^{\calO}=\{g\in\calS_0:D_{\calO}(g;\rho)=0\}$, and suppose $D_{\calO}(g;\rho)\ge\gamma_{\calO}>0$ for all remaining elements. Under the bounded-snapshot model, there is a universal constant $C$ such that
\begin{equation}
M\ge\max\left\{2,\ C\frac{R_{\calO}^2}{\gamma_{\calO}}\log\frac{2m|\calS_0|}{\delta}\right\}
\label{eq:shadow_sample}
\end{equation}
ensures that $\widehat{\calS}^{\calO}=\{g:\widehat D_{\calO}(g)\le\gamma_{\calO}/2\}$ equals $\calS_\star^{\calO}$ with probability at least $1-\delta$. The dependence $\gamma_{\calO}^{-1}\log(1/\delta)$ is necessary in general, even for one qubit and one candidate, for $0<\gamma_{\calO}\le1$ and $0<\delta<1/4$.
\end{theorem}

\begin{proof}
Write $\mu_j^{(g)}=\E X_{\ell,j}^{(g)}$, $\bar X_j^{(g)}=M^{-1}\sum_\ell X_{\ell,j}^{(g)}$ and $s_{j,g}^2=(M-1)^{-1}\sum_\ell(X_{\ell,j}^{(g)}-\bar X_j^{(g)})^2$. The exact identity
\begin{equation}
\widehat q_j^{(g)}=(\bar X_j^{(g)})^2-\frac{s_{j,g}^2}{M}
\label{eq:u_mean_variance}
\end{equation}
exposes the smaller error at a zero mean. Hoeffding's inequality \cite{Hoeffding1963} and a union bound give $|\bar X_j^{(g)}-\mu_j^{(g)}|\le t_M$ simultaneously, where $t_M=R_{\calO}\sqrt{2\log(2m|\calS_0|/\delta)/M}$. Also $0\le s_{j,g}^2/M\le R_{\calO}^2/(M-1)$. Applying the Euclidean triangle inequality to the vector of means gives
\begin{equation}
(\sqrt{D_{\calO}(g)}-t_M)_+^2-\frac{R_{\calO}^2}{M-1}
\le\widehat D_{\calO}(g)
\le(\sqrt{D_{\calO}(g)}+t_M)^2.
\label{eq:shadow_signal_bound}
\end{equation}
Choose the universal constant so that $t_M\le\sqrt{\gamma_{\calO}}/8$ and $R_{\calO}^2/(M-1)\le\gamma_{\calO}/16$. At a null element the upper bound is below $\gamma_{\calO}/2$; at any alternative the lower bound exceeds it.

For the information-theoretic lower bound \cite{BoucheronLugosiMassart2013}, use the commuting states $\rho_0=I/2$ and $\rho_a=(I+aZ)/2$, probe $Z$ and candidate $U_g=X$, so the alternative discrepancy is $4a^2=\gamma_{\calO}$. Any measurement on $M$ copies is a stochastic postprocessing of the classical diagonal samples. The one-copy relative entropy is at most $a^2$, by the relative-entropy bound by the logarithm of one plus chi-square divergence. A decision with each error at most $\delta$ has binary relative entropy at least $(1-2\delta)\log((1-\delta)/\delta)$, by the log-sum inequality. Hence $Ma^2\ge(1-2\delta)\log((1-\delta)/\delta)$, proving the claimed lower bound even with collective measurement access.
\end{proof}

The gain in \eqref{eq:shadow_sample} comes from resolving a zero squared discrepancy. Estimating every discrepancy to a uniform additive accuracy would hide this structure and pay an inverse-square-gap rate. The result is optimal in the gap at fixed snapshot scale; it does not assert optimal dependence on the probe locality or dictionary size.

Set $\widehat V_g=m^{-1}\sum_j s_{j,g}^2$ and $Q_g=\widehat D_{\calO}(g)+\widehat V_g/M=m^{-1}\sum_j(\bar X_j^{(g)})^2$. On the same event as \eqref{eq:shadow_signal_bound},
\begin{equation}
D_{\calO}(g)\in\left[(\sqrt{Q_g}-t_M)_+^2,(\sqrt{Q_g}+t_M)^2\right]
\quad\text{for every }g\in\calS_0.
\label{eq:observable_intervals}
\end{equation}
These intervals remain simultaneous for any tolerance selected from the data: declare an element below tolerance only when its upper endpoint is below it, and above tolerance only when its lower endpoint exceeds it. For a fixed tolerance $\tau$, exact recovery of $\{g:D_{\calO}(g)\le\tau\}$ separated from its complement by $\gamma_{\calO}$ follows by thresholding the $U$-statistic at $\tau+\gamma_{\calO}/2$ with $M\ge C R_{\calO}^2(\tau+\gamma_{\calO})\gamma_{\calO}^{-2}\log(2m|\calS_0|/\delta)$ and $M\ge2$, by \eqref{eq:shadow_signal_bound}. Choosing the scientific relevance of the tolerance requires the task loss; its statistical validity follows from the simultaneous event.

\begin{proposition}[Shadow constants for local dictionaries]
\label[proposition]{prop:shadow_constants}
Let $\calO$ consist of Pauli observables of weight at most $k$ with $\|O\|_{\mathrm{op}}\le1$, measured with random single-qubit Pauli shadows. Assume that for every $O\in\calO$ and $g\in\calS_0$ the conjugated observable $U_g^\dagger O U_g$ is a Pauli observable of weight at most $k'$, uniformly over $\calO\times\calS_0$. Then the relevant squared Pauli-shadow norms of $O$ and $U_g^\dagger O U_g$ are bounded by $3^k$ and $3^{k'}$, respectively. For $0<\varepsilon\le1$, median-of-means estimates of all expectation values entering the plug-in discrepancy achieve additive accuracy $\varepsilon$ per expectation and failure probability at most $\delta$ using
\begin{equation}
M=O\!\left(\frac{3^{\max(k,k')}}{\varepsilon^2}
\log\frac{2|\calO|\,|\calS_0|}{\delta}\right)
\end{equation}
shadow snapshots. For the transverse-field Ising dictionary used in \cref{subsec:tfim-baseline}, the probes $\{X_i\}$ and $\{Z_iZ_{i+1}\}$ have $k\le2$, and spin flips or translations preserve the Pauli weight, so the sample cost is logarithmic in the dictionary size at fixed accuracy.
\end{proposition}

\begin{proof}
For random single-qubit Pauli shadows, the squared shadow norm of a Pauli string of weight $w$ is at most $3^w$ \cite{Huang2020ClassicalShadows}. The same bound applies to the conjugated observables by the assumed weight bound $k'$. The median-of-means simultaneous-prediction theorem applied to the finite family
\[
\{O,\;U_g^\dagger O U_g:O\in\calO,\;g\in\calS_0\}
\]
and union-bounded over at most $2|\calO|\,|\calS_0|$ observables gives the stated scaling. Clip each estimate to $[-1,1]$ before forming differences. Each difference then has error at most $2\varepsilon$ and magnitude at most two, so the resulting plug-in discrepancy has error at most $8\varepsilon$.
\end{proof}

For Pauli snapshots, $R_{\calO}\le3^k+3^{k'}$ is an almost-sure bound on a discrepancy snapshot, so \cref{thm:shadow_recovery} gives a directly applicable local recovery guarantee. The variance-based median-of-means bound in \cref{prop:shadow_constants} concerns linear expectation estimates and can have better locality constants; it is a different estimator guarantee. Under gate-independent, time-stationary measurement noise, calibrated robust shadows \cite{Chen2021RobustShadows,Koh2022RobustShadows} are unbiased when the invertible measurement channel is known exactly; the cost is model-dependent rather than a single universal factor, since for a local Pauli observable the ordinary shadow norm is replaced by its noise-calibrated counterpart, typically amplified by inverse Pauli fidelities on the support.  For an estimated calibration channel, the confidence event must additionally include calibration uncertainty and its induced bias; a variance rescaling alone is insufficient.

\subsection{Continuous discovery through a Lie-algebraic null space}
\label{subsec:lie_algebra_null}

The finite-dictionary theorem above has only logarithmic dependence on $\abs{\calS_0}$.  Compact Lie groups require a local formulation because exhaustive testing of all group elements is unavailable.  For continuous symmetries, the corresponding local problem is to recover the infinitesimal directions that are invisible to the chosen observable family.

For continuous discovery, $G$ is a compact matrix Lie group. We use Hermitian generators of $e^{-itL}$, with bracket $i[L,L\prime]$; the kernel of any original nonfaithful representation is known separately. Fix a real basis of its represented algebra and hold the coefficient normalization fixed throughout the spectral comparison:
\begin{equation}
\mathfrak g=\Span_{\R}\{L_1,\ldots,L_K\},
\qquad
L_{\bm\theta}:=\sum_{a=1}^K\theta_aL_a.
\end{equation}
The one-parameter perturbation generated by $L_{\bm\theta}$ acts as
\begin{equation}
\rho(t)=e^{-itL_{\bm\theta}}\rho e^{itL_{\bm\theta}}.
\end{equation}
For each probe $O_j$, its infinitesimal expectation change is
\begin{equation}
\left.\frac{d}{dt}\Tr(O_j\rho(t))\right|_{t=0}
=\Tr\bigl(i[L_{\bm\theta},O_j]\rho\bigr).
\label{eq:infinitesimal_expectation}
\end{equation}

\begin{definition}[Lie-algebraic observable covariance]
\label[definition]{def:lie_cov}
Define the symmetric positive semidefinite matrix $\mathbf M_{\mathfrak g}(\rho;\calO)\in\R^{K\times K}$ by
\begin{equation}
(\mathbf M_{\mathfrak g})_{ab}
:=\frac1m\sum_{j=1}^m
\Tr\bigl(i[L_a,O_j]\rho\bigr)
\Tr\bigl(i[L_b,O_j]\rho\bigr).
\label{eq:lie_cov}
\end{equation}
The infinitesimal observable-invisible space is
\begin{equation}
\mathcal N_{\calO}(\rho):=
\{\bm\theta\in\R^K:\Tr(i[L_{\bm\theta},O_j]\rho)=0
\text{ for all }j\le m\}.
\end{equation}
\end{definition}

By construction, $\mathcal N_{\calO}(\rho)=\ker\mathbf M_{\mathfrak g}(\rho;\calO)$. Under the invariant-span condition of \cref{thm:probe_stabilizer}, this is the Lie algebra of the invisible subgroup. For a general restricted family it certifies only first-order invisibility; bracket closure does not remove the possibility of a higher-order change of probe expectations.

Let
\begin{equation}
A_{a,j}:=i[L_a,O_j],
\end{equation}
which is Hermitian when $L_a$ and $O_j$ are Hermitian.  Given classical-shadow snapshots $\widehat\rho_\ell$, define
\begin{equation}
\widehat M_{ab}:=
\frac1m\sum_{j=1}^m
\frac{1}{M(M-1)}\sum_{\ell\ne r}
\Tr(A_{a,j}\widehat\rho_\ell)\Tr(A_{b,j}\widehat\rho_r).
\label{eq:lie_u_stat}
\end{equation}
Then $\widehat{\mathbf M}$ is real symmetric and unbiased entrywise. It can have negative eigenvalues at finite sample size, so the recovery rule uses the entire low-eigenvalue cluster. If the population matrix is zero, a confidence radius $\varepsilon$ certifies all eigenvalues within $[-\varepsilon,\varepsilon]$; distinguishing this case from a weak nonzero response requires a specified positive-gap alternative.

\begin{theorem}[Stable recovery of infinitesimal invisible directions]
\label{thm:lie_recovery}
Assume all shadow variables associated with $A_{a,j}$ have absolute value at most $R_{\mathfrak g}$ almost surely.  Let $\gamma_{\mathbf M}>0$ be the smallest nonzero eigenvalue of $\mathbf M_{\mathfrak g}(\rho;\calO)$.  There is a universal constant $C$ such that, for any $0<\varepsilon\le\gamma_{\mathbf M}/4$, the condition
\begin{equation}
M\ge\max\!\left\{2,\ C\,\frac{R_{\mathfrak g}^4K^2}{\varepsilon^2}
\log\frac{2K^2}{\delta}\right\}
\label{eq:lie_sample}
\end{equation}
implies
\begin{equation}
\|\widehat{\mathbf M}-\mathbf M_{\mathfrak g}\|_{\mathrm{op}}\le \varepsilon
\end{equation}
with probability at least $1-\delta$.  Consequently, the empirical subspace spanned by eigenvectors of $\widehat{\mathbf M}$ with eigenvalues at most $\gamma_{\mathbf M}/2$ has the same dimension as $\mathcal N_{\calO}(\rho)$ and satisfies a Davis--Kahan stability bound
\begin{equation}
\|\sin\Theta(\widehat{\mathcal N}_{\calO},\mathcal N_{\calO})\|_{\mathrm{op}}
\le \frac{2\varepsilon}{\gamma_{\mathbf M}}.
\end{equation}
\end{theorem}

\begin{proof}
For any $\bm\theta\in\R^K$,
\begin{equation}
\bm\theta^\top\mathbf M_{\mathfrak g}\bm\theta
=\frac1m\sum_{j=1}^m
\left(\Tr(i[L_{\bm\theta},O_j]\rho)\right)^2,
\end{equation}
so the quadratic form vanishes exactly on $\mathcal N_{\calO}(\rho)$.  The symmetrized two-snapshot kernel of each entry in \eqref{eq:lie_u_stat}, already averaged over the probes, lies in $[-R_{\mathfrak g}^2,R_{\mathfrak g}^2]$. Averaging over all permutations expresses the $U$-statistic as an average of means over $\lfloor M/2\rfloor$ independent pairs. Jensen's inequality for the exponential and Hoeffding's bound for each pair mean give a tail $2\exp[-\lfloor M/2\rfloor t^2/(2R_{\mathfrak g}^4)]$. Union bounding over the $K(K+1)/2$ entries and using $\|E\|_{\mathrm{op}}\le K\max_{ab}|E_{ab}|$ proves \eqref{eq:lie_sample}.  Weyl's inequality separates the empirical small-eigenvalue cluster from the nonzero spectrum, and the Davis--Kahan theorem gives the stated subspace perturbation bound \cite{DavisKahan1970}.
\end{proof}

Thus continuous discovery is reduced to a $K$-dimensional spectral problem in the chosen generator basis.  The guarantee is polynomial in the Lie-algebra dimension and local in the measured observable family, rather than a discretization of the full group manifold.

The restricted-probe statement becomes an exact stabilizer statement when the probes are informationally complete for the state variation generated by $\mathfrak g$.

\begin{remark}[Reality classes of probes]
\label[remark]{rem:reality_classes}
Call a Hermitian operator \emph{real} if its matrix in the computational basis is real, and \emph{imaginary} if that matrix is purely imaginary; a Pauli string is real exactly when it contains an even number of $Y$ factors and imaginary otherwise, so the Pauli basis splits into these two classes.  If $\rho$ is real, as for a real choice of ground state of a real Hamiltonian, then $\Tr(i[L,O]\rho)=0$ whenever $L$ and $O$ lie in the same class, since $[L,O]$ is then real antisymmetric and a real antisymmetric matrix has zero trace against a real symmetric one.  An all-real probe family therefore makes every real generator invisible at first order, and $\ker\mathbf M_{\mathfrak g}$ contains such generators whether or not they are symmetries.  To resolve generators from both classes on a real state, the probe family must provide the corresponding opposite-class responses.  For the periodic $L=8$ XXZ chain $H_{\mathrm{XXZ}}=\sum_i(X_iX_{i+1}+Y_iY_{i+1}+\Delta Z_iZ_{i+1})$ at $\Delta=1.5$ with candidates $\{M_x,M_y,M_z\}$, the probes $\{Z_i,Z_iZ_{i+1},X_iX_{i+1}\}$ give $\mathbf M_{\mathfrak g}=0$ to machine precision (kernel dimension three); adding $X_iZ_{i+1}$ detects $M_y$ (kernel dimension two); adding the imaginary probes $Y_iZ_{i+1}$ and $X_iY_{i+1}$ detects $M_x$ as well and leaves exactly the $U(1)$ generator $M_z$ in the kernel.  The informationally complete covariance of \cref{prop:ic_stabilizer} below is unaffected, because a complete Hermitian basis contains both classes.
\end{remark}

\begin{proposition}[Population recovery of the stabilizer Lie algebra]
\label[proposition]{prop:ic_stabilizer}
Let $\{B_\mu\}_{\mu=1}^{d^2}$ be a Hilbert--Schmidt orthonormal Hermitian basis of $\Herm(\calH)$, where $d=\dim\calH$, and define
\begin{equation}
(\mathbf M_{\mathrm{IC}})_{ab}
:=\sum_{\mu=1}^{d^2}
\Tr\bigl(i[L_a,B_\mu]\rho\bigr)
\Tr\bigl(i[L_b,B_\mu]\rho\bigr).
\label{eq:ic_lie_cov}
\end{equation}
Then
\begin{equation}
\bm\theta^\top\mathbf M_{\mathrm{IC}}\bm\theta
=\left\|i[\rho,L_{\bm\theta}]\right\|_{\mathrm{HS}}^2.
\label{eq:ic_cov_hs}
\end{equation}
Consequently,
\begin{equation}
\ker\mathbf M_{\mathrm{IC}}
=\{\bm\theta\in\R^K:[L_{\bm\theta},\rho]=0\}.
\label{eq:ic_kernel}
\end{equation}
Under the chosen generator-basis identification, this kernel is the Lie algebra of the stabilizer subgroup
\begin{equation}
G_\rho:=\{g\in G:U_g\rho U_g^\dagger=\rho\}
\end{equation}
inside the candidate group.
\end{proposition}

\begin{proof}
For Hermitian $B_\mu$,
\begin{equation}
\Tr\bigl(i[L_a,B_\mu]\rho\bigr)
=\Tr\bigl(B_\mu\, i[\rho,L_a]\bigr).
\end{equation}
The coefficients in \eqref{eq:ic_lie_cov} are therefore the Hilbert--Schmidt coordinates of $i[\rho,L_a]$ in an orthonormal Hermitian operator basis.  Parseval's identity gives \eqref{eq:ic_cov_hs}.  The kernel condition is exactly $[L_{\bm\theta},\rho]=0$.  If two Hermitian generators $L$ and $L'$ commute with $\rho$, then Jacobi's identity gives $[[L,L'],\rho]=0$, so the kernel is closed under the Hermitian-generator bracket $i[L,L']$.  It is precisely the infinitesimal stabilizer of the conjugation action $U\rho U^\dagger$ restricted to the candidate group $G$.
\end{proof}

\begin{corollary}[Finite-sample stabilizer recovery under complete probes]
\label[corollary]{cor:ic_stabilizer_recovery}
Let $\widehat{\mathbf M}_{\mathrm{IC}}$ be an unbiased second-order shadow estimator of \eqref{eq:ic_lie_cov}.  Assume the symmetrized two-snapshot kernel of each entry, including the sum over the complete probe frame, is bounded in absolute value by $R_{\mathrm{IC}}$.  Let $\gamma_{\mathrm{IC}}>0$ be the smallest nonzero eigenvalue of $\mathbf M_{\mathrm{IC}}$.  There is a universal constant $C$ such that, for $0<\varepsilon\le\gamma_{\mathrm{IC}}/4$, the condition
\begin{equation}
M\ge\max\!\left\{2,\ C\,\frac{R_{\mathrm{IC}}^2K^2}{\varepsilon^2}
\log\frac{2K^2}{\delta}\right\}
\label{eq:ic_lie_sample}
\end{equation}
implies $\|\widehat{\mathbf M}_{\mathrm{IC}}-\mathbf M_{\mathrm{IC}}\|_{\mathrm{op}}\le\varepsilon$ with probability at least $1-\delta$.  The empirical eigenspace below $\gamma_{\mathrm{IC}}/2$ then has the stabilizer dimension and obeys
\begin{equation}
\|\sin\Theta(\widehat{\mathfrak g}_\rho,\mathfrak g_\rho)\|_{\mathrm{op}}
\le \frac{2\varepsilon}{\gamma_{\mathrm{IC}}}.
\end{equation}
\end{corollary}

\begin{proof} Apply the bounded-kernel argument in \cref{thm:lie_recovery} to the complete-frame sum, then use \cref{prop:ic_stabilizer} to identify its kernel. \end{proof}

Complete probes identify the full state stabilizer within the candidate Lie group. Local probes instead identify the projected tangent response. Full-state recovery from them is possible when the local marginal map is injective on the candidate tangent space; the next identity makes that condition explicit.

\begin{proposition}[Local informational completeness for local generators]
\label[proposition]{prop:local_ic}
Let \(\calH=(\mathbb C^2)^{\otimes n}\), and let \(\mathcal V_r\subset\Herm(\calH)\) be the real Hilbert--Schmidt subspace spanned by Pauli strings of weight at most \(r\).  Let \(\Pi_r\) be the Hilbert--Schmidt projection onto \(\mathcal V_r\), and let \(\{P_\mu\}_{\mu=1}^{d_r}\) be its Hilbert--Schmidt-normalized Pauli-string basis.  Define
\begin{equation}
(\mathbf M_r)_{ab}:=
\sum_{\mu=1}^{d_r}
\Tr\bigl(i[L_a,P_\mu]\rho\bigr)
\Tr\bigl(i[L_b,P_\mu]\rho\bigr).
\label{eq:local_lie_cov}
\end{equation}
Then
\begin{equation}
\bm\theta^\top\mathbf M_r\bm\theta
=
\left\|\Pi_r\bigl(i[\rho,L_{\bm\theta}]\bigr)\right\|_{\mathrm{HS}}^2.
\label{eq:local_ic_projection}
\end{equation}
Consequently, \(\ker\mathbf M_r\) is exactly the tangent space of generator directions whose infinitesimal action is invisible to all \(r\)-local observables.  If each basis generator \(L_a\) is supported on at most \(k\) sites, each scalar \(\Tr(i[L_a,P_\mu]\rho)\) depends only on the reduced density matrix of \(\rho\) on \(\operatorname{supp}(L_a)\cup\operatorname{supp}(P_\mu)\), whose size is at most \(k+r\).  For generators that are sums of \(k\)-local terms, the same conclusion applies term by term and the sample cost acquires the corresponding sparsity factor.  For fixed \(r\), the number of Pauli probes is
\begin{equation}
d_r=\sum_{s=0}^r 3^s\binom{n}{s}=O(n^r),
\end{equation}
so the probe count is polynomial in system size at fixed locality. Total recovery cost also depends on the inverse spectral gap, generator coefficients, and requested accuracy.
\end{proposition}

\begin{proof}
For each \(P_\mu\),
\begin{equation}
\Tr\bigl(i[L_a,P_\mu]\rho\bigr)
=\Tr\bigl(P_\mu i[\rho,L_a]\bigr).
\end{equation}
Summing squared Hilbert--Schmidt coordinates over an orthonormal basis of \(\mathcal V_r\) gives \eqref{eq:local_ic_projection}.  The kernel statement follows immediately.  If \(L_a\) and \(P_\mu\) have supports \(A\) and \(B\), then \([L_a,P_\mu]\) is supported on \(A\cup B\).  Its expectation therefore depends only on the reduced density matrix on that union.  A termwise \(k\)-local generator follows by linearity.
\end{proof}

By \cref{prop:local_ic}, the local and full stabilizer kernels agree precisely when $\Pi_r$ is injective on $\{i[\rho,L]:L\in\mathfrak g\}$. Global identification of a state within a nonlinear model class does not by itself provide this differential condition. When the probe span is invariant under the candidate group, \cref{thm:probe_stabilizer} already integrates the local kernel to the projected-state stabilizer, whether or not it agrees with the full-state stabilizer.

\subsection{Cyclic finite-memory benchmark}
\label{subsec:qfa_benchmark}

\Cref{thm:shadow_recovery} is particularly transparent for a finite-memory quantum process with cyclic transition symmetry.  Let $C_m=\langle r:r^m=e\rangle$ act on a memory basis $\{|a\rangle:a\in\mathbb Z_m\}$ by
\begin{equation}
U_{r^\ell}\ket a=\ket{a+\ell\;\mathrm{mod}\;m}.
\end{equation}
This includes the memory-register part of simple quantum finite automata and modular periodic-language processes \cite{MooreCrutchfield2000,AmbainisFreivalds1998}.  For the local projector probes $O_a=\ket a\!\bra a$, the discrepancy becomes
\begin{equation}
D_{\calO}(r^\ell;p)=\frac1m\sum_{a\in\mathbb Z_m}\bigl(p_a-p_{a-\ell}\bigr)^2,
\label{eq:cyclic_disc}
\end{equation}
where $p_a=\langle a|\rho|a\rangle$ is the observed memory distribution.  Thus cyclic-symmetry discovery is equivalent to recovering the shifts that preserve the observed periodic occupation pattern.

\begin{proposition}[Recovery of a cyclic period from low-dimensional probes]
\label[proposition]{prop:cyclic_recovery}
Assume $p$ is invariant exactly under a proper subgroup $K<C_m$ and define
\begin{equation}
\gamma_C:=\min_{r^\ell\notin K}D_{\calO}(r^\ell;p)>0.
\end{equation}
Let $\widehat p$ be the empirical distribution from $M$ independent measurements in the memory basis and let $\widehat D_{\calO}$ be obtained from \eqref{eq:cyclic_disc} with $p$ replaced by $\widehat p$.  There is a universal constant $C$ such that
\begin{equation}
M\ge C\,\frac{1}{m\gamma_C}\log\frac{2}{\delta}
\end{equation}
ensures that thresholding $\widehat D_{\calO}(r^\ell)$ at $\gamma_C/2$ recovers $K$ with probability at least $1-\delta$.
\end{proposition}

\begin{proof}
For categorical samples, $\E\|\widehat p-p\|_2\le M^{-1/2}$. Changing one observation changes this norm by at most $\sqrt2/M$, so the bounded-differences inequality \cite{BoucheronLugosiMassart2013} gives $\|\widehat p-p\|_2\le(1+\sqrt{\log(1/\delta)})/\sqrt M$ with probability at least $1-\delta$. Every cyclic shift is an isometry. Consequently, simultaneously for all shifts,
\begin{equation}
\left|\sqrt{D_{\calO}(r^\ell;\widehat p)}-\sqrt{D_{\calO}(r^\ell;p)}\right|
\le\frac{2}{\sqrt m}\|\widehat p-p\|_2.
\end{equation}
The sample condition makes the right-hand side at most $\sqrt{\gamma_C}/4$. Squaring then separates the null values, at most $\gamma_C/16$, from the alternatives, at least $9\gamma_C/16$, across the threshold $\gamma_C/2$.
\end{proof}

For a concrete reproducible illustration, take \(m=12\), \(K=\langle r^4\rangle=\{e,r^4,r^8\}\), and the noiseless distribution \(p^{(0)}\) that is uniform on the invariant residue class \(\{0,4,8\}\) and zero elsewhere.  \Cref{fig:cyclic_recovery} plots the recovery threshold for \(p^{(\epsilon)}=(1-\epsilon)p^{(0)}+\epsilon\,\mathbf 1/12\), and recovery uses the threshold \(\gamma_C/2\) from Proposition~\ref{prop:cyclic_recovery}.

\begin{figure}[t]
\centering
\begin{tikzpicture}
\begin{axis}[
    width=0.78\textwidth,
    height=0.46\textwidth,
    xlabel={number of memory-basis samples $M$},
    ylabel={generator recovery success},
    xmin=15,xmax=2200,
    ymin=0,ymax=1.04,
    xmode=log,
    xtick={20,50,100,200,500,1000,2000},
    xticklabels={20,50,100,200,500,1000,2000},
    ytick={0,0.25,0.5,0.75,1},
    grid=both,
    legend style={draw=none,fill=none,at={(0.03,0.03)},anchor=south west},
    legend cell align=left]
\addplot[mark=*] coordinates {(20,0.4480) (50,0.8900) (100,0.9835) (200,1.0000) (500,1.0000) (1000,1.0000) (2000,1.0000)};
\addlegendentry{moderate gap, $\epsilon=0.5$}
\addplot[mark=square*] coordinates {(20,0.0010) (50,0.0530) (100,0.2700) (200,0.7265) (500,0.9720) (1000,0.9950) (2000,1.0000)};
\addlegendentry{small gap, $\epsilon=0.8$}
\end{axis}
\end{tikzpicture}
\caption{Illustrative cyclic-symmetry recovery curve for $m=12$ and the hidden subgroup \(K=\langle r^4\rangle\).  Each point is the empirical success probability over 2000 multinomial simulations from \(p^{(\epsilon)}\), using the threshold rule in Proposition~\ref{prop:cyclic_recovery}.  }
\label{fig:cyclic_recovery}
\end{figure}
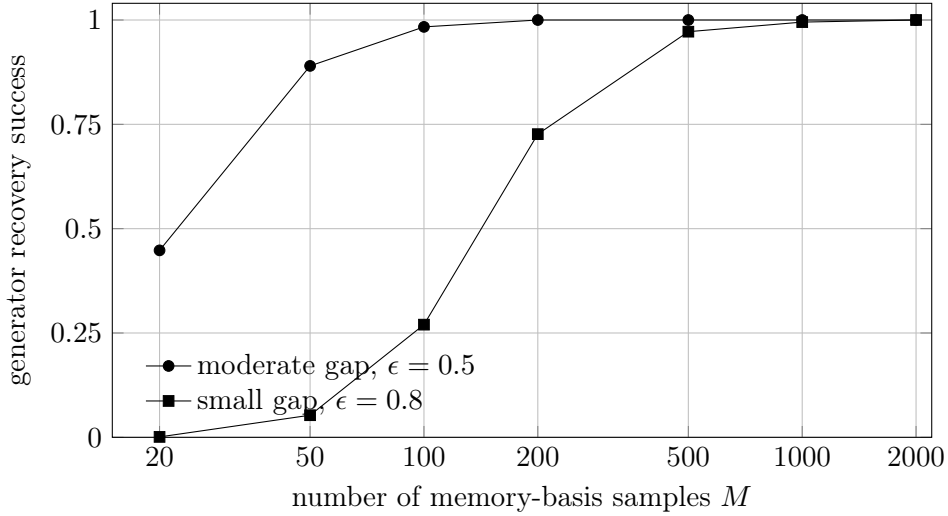

\section{Task-level discovery for supervised quantum learning}
\label{sec:task_discovery}

Observable discovery concerns the information carried by a state and its probes. Supervised symmetry concerns how the input and label transform together. Even when the input ensemble is invariant, an order-parameter label may require a nontrivial output action. Maximum mean discrepancy (MMD) tests this joint law under classical sample access; a classical--quantum mean-state discrepancy tests the information retained by a quantum encoding.

\subsection{Maximum mean discrepancy on the joint distribution}

Let $P$ be the distribution of $Z=(X,Y)\in\calZ$.  A task-level symmetry is an element $g\in G$ such that
\begin{equation}
T_g\push P=P,
\qquad
T_g(x,y)=(g\cdot x,\rho_{\calY}(g)y).
\end{equation}
Let $k:\calZ\times\calZ\to\R$ be a bounded measurable positive-definite kernel with reproducing-kernel Hilbert space $\calH_k$.  The maximum mean discrepancy is
\begin{equation}
\MMD(P,Q):=\sup_{\|f\|_{\calH_k}\le 1}\abs{\E_{Z\sim P}f(Z)-\E_{W\sim Q}f(W)}.
\end{equation}
When $k$ is characteristic, $\MMD(P,Q)=0$ if and only if $P=Q$ \cite{Gretton2012KernelTwoSample}.

Given samples $z_1,\ldots,z_n\sim P$, set $z_i^{(g)}:=T_gz_i$ and define the empirical mean-embedding discrepancy
\begin{equation}
\widehat{\MMD}_g:=\left\|\frac1n\sum_{i=1}^n k(\cdot,z_i)-\frac1n\sum_{i=1}^n k(\cdot,z_i^{(g)})\right\|_{\calH_k}.
\label{eq:mmd_emp}
\end{equation}
Equivalently, its square is computed by the usual kernel trick.

For large dictionaries it is often useful to pre-screen the candidates before performing the MMD validation.  Suppose a feature map $\Phi:\calZ\to\R^p$ is available and that the transformed features are represented by known orthogonal matrices $R_\Phi(g)$, so that $\Phi(T_g z)=R_\Phi(g)\Phi(z)$ in the ideal feature model.  Let
\begin{equation}
\Sigma_P:=\E[(\Phi(Z)-\E\Phi(Z))(\Phi(Z)-\E\Phi(Z))^\top].
\end{equation}

\begin{proposition}[Covariance-commutator pre-screen]
\label[proposition]{prop:commutator_prescreen}
If $T_g\push P=P$, $\Phi(T_g z)=R_\Phi(g)\Phi(z)$, and $R_\Phi(g)$ is orthogonal, then
\begin{equation}
R_\Phi(g)\Sigma_P=\Sigma_P R_\Phi(g).
\end{equation}
Consequently, the empirical residual
\begin{equation}
r_{\Phi}(g):=\left\|R_\Phi(g)\widehat\Sigma-\widehat\Sigma R_\Phi(g)\right\|_{\mathrm F}
\label{eq:commutator_residual}
\end{equation}
is a necessary-condition score for task-level symmetry.
\end{proposition}

\begin{proof}
Let $m_P:=\E\Phi(Z)$.  Invariance and feature equivariance give $R_\Phi(g)m_P=m_P$.  Therefore
\begin{equation}
R_\Phi(g)\Sigma_P R_\Phi(g)^\top
=\E[(\Phi(T_gZ)-m_P)(\Phi(T_gZ)-m_P)^\top]
=\Sigma_P.
\end{equation}
Since $R_\Phi(g)$ is orthogonal, this is equivalent to the commutation relation.
\end{proof}

The residual depends on the chosen features and can fail under approximate feature equivariance, so it screens a large candidate set rather than deciding it.

\begin{theorem}[Task-level dictionary recovery]
\label{thm:mmd_recovery}
Assume $0\le k(z,z)\le \kappa_k$ for all $z$ and let
\begin{equation}
\calS_\star^{P}:=\{g\in\calS_0:T_g\push P=P\}.
\end{equation}
Suppose the non-symmetric dictionary elements have a positive MMD margin
\begin{equation}
\MMD(P,T_g\push P)\ge \gamma_P>0
\qquad\text{for all }g\in\calS_0\setminus\calS_\star^P.
\end{equation}
There exists a universal constant $C$ such that, if
\begin{equation}
n\ge C\,\frac{\kappa_k}{\gamma_P^2}\log\frac{2\abs{\calS_0}}{\delta},
\end{equation}
then the threshold rule
\begin{equation}
\widehat{\calS}^{P}:=\{g\in\calS_0:\widehat{\MMD}_g\le \gamma_P/2\}
\end{equation}
recovers $\calS_\star^P$ with probability at least $1-\delta$.
\end{theorem}

\begin{proof}
Write $\xi_i^g=k(\cdot,z_i)-k(\cdot,T_gz_i)$. These paired feature differences are independent across $i$ and satisfy $\|\xi_i^g\|\le2\sqrt{\kappa_k}$. Their mean error has expected norm at most $2\sqrt{\kappa_k/n}$, and replacing one example changes that norm by at most $4\sqrt{\kappa_k}/n$. Bounded differences \cite{BoucheronLugosiMassart2013} and a union bound therefore give, simultaneously for all candidates,
\begin{equation}
\left\|\frac1n\sum_i\xi_i^g-\E\xi_i^g\right\|
\le 2\sqrt{\frac{\kappa_k}{n}}+\sqrt{\frac{8\kappa_k\log(|\calS_0|/\delta)}{n}}.
\end{equation}
The reverse triangle inequality bounds $|\widehat{\MMD}_g-\MMD(P,T_g\push P)|$ by this expression. Choose $C$ so that it is at most $\gamma_P/4$. Null elements then have empirical discrepancy at most $\gamma_P/4$, and alternatives have discrepancy at least $3\gamma_P/4$, which separates them strictly across the threshold.

\end{proof}

\subsection{Quantum-native task discrepancy through classical--quantum states}
\label{subsec:cq_states}

Kernel MMD is a flexible distributional validation tool, but it depends on the chosen reproducing-kernel Hilbert space.  When the supervised inputs are quantum states, one can instead test task-level symmetry directly in the Hilbert space carrying the data.

Assume a supervised example consists of a quantum input state $\rho_X$ and a discrete label $Y$.  Embed the label in an auxiliary register with orthonormal basis $\{\ket y:y\in\calY\}$ and define the classical--quantum (CQ) task state
\begin{equation}
\rho_{\calZ}:=\E_{(X,Y)\sim P}\bigl[\rho_X\otimes\ket Y\!\bra Y\bigr].
\label{eq:cq_state}
\end{equation}
A candidate symmetry acts on this joint register as
\begin{equation}
\mathbb U_g:=U_g\otimes V_g,
\end{equation}
where the classical label action is implemented by $V_g\ket y=\ket{\rho_{\calY}(g)y}$. Thus a numerical sign action on $y\in\{+1,-1\}$ swaps the two label basis states. More general unitary $V_g$ still define a joint-register symmetry test, but represent the classical task transformation only when the encoding intertwines the two actions.

\begin{definition}[CQ task discrepancy]
\label[definition]{def:cq_discrepancy}
The quantum task discrepancy is the squared Hilbert--Schmidt distance
\begin{equation}
D_{\mathrm{CQ}}(g):=
\frac12\left\|\rho_{\calZ}-\mathbb U_g\rho_{\calZ}\mathbb U_g^\dagger\right\|_{\mathrm{HS}}^2
=
\Tr(\rho_{\calZ}^2)
-\Tr\!\left(\rho_{\calZ}\mathbb U_g\rho_{\calZ}\mathbb U_g^\dagger\right).
\label{eq:cq_discrepancy}
\end{equation}
\end{definition}

The equality in \eqref{eq:cq_discrepancy} uses unitary invariance of purity.  The condition $D_{\mathrm{CQ}}(g)=0$ is equivalent to invariance of the embedded CQ state.  This may be weaker than full invariance of the original distribution if the mean-embedding map $P\mapsto\rho_{\calZ}$ is not injective on the relevant distribution class, so the statistic should be interpreted as a task-level test at the chosen quantum encoding.

\begin{proposition}[Two-copy estimator of the CQ task discrepancy]
\label[proposition]{prop:cq_two_copy}
Suppose independent preparations of the CQ state \eqref{eq:cq_state} are available and $\mathbb U_g$ can be applied to one copy, and let $\mathbb S$ denote the swap operator between two copies of the joint input--label register.  Then
\begin{equation}
\Tr(\rho_{\calZ}^2)=\Tr\!\left(\mathbb S\,\rho_{\calZ}\otimes\rho_{\calZ}\right),
\qquad
\Tr\!\left(\rho_{\calZ}\mathbb U_g\rho_{\calZ}\mathbb U_g^\dagger\right)
=\Tr\!\left(\mathbb S\,\rho_{\calZ}\otimes\mathbb U_g\rho_{\calZ}\mathbb U_g^\dagger\right),
\label{eq:cq_swap_overlaps}
\end{equation}
so $D_{\mathrm{CQ}}(g)$ has an unbiased estimator from two independent SWAP measurements, with variance at most $2/M$ when $M$ shots are used per overlap.  Moreover, let
\begin{equation}
\calS_\star^{\mathrm{CQ}}:=\{g\in\calS_0:D_{\mathrm{CQ}}(g)=0\}
\end{equation}
and suppose $D_{\mathrm{CQ}}(g)\ge\gamma_{\mathrm{CQ}}$ for all $g\in\calS_0\setminus\calS_\star^{\mathrm{CQ}}$.  There is a universal constant $C$ such that, if
\begin{equation}
M\ge C\,\gamma_{\mathrm{CQ}}^{-2}\log\frac{2\abs{\calS_0}}{\delta}
\end{equation}
shots are used per overlap estimate and per tested dictionary element, then the threshold rule
\begin{equation}
\widehat{\calS}^{\mathrm{CQ}}
:=\{g\in\calS_0:\widehat D_{\mathrm{CQ}}(g)\le\gamma_{\mathrm{CQ}}/2\}
\end{equation}
recovers $\calS_\star^{\mathrm{CQ}}$ with probability at least $1-\delta$.
\end{proposition}

\begin{proof}
For any density operators $\rho$ and $\sigma$ on the same Hilbert space, $\Tr(\mathbb S\rho\otimes\sigma)=\Tr(\rho\sigma)$; applying this with $\sigma=\rho_{\calZ}$ and with $\sigma=\mathbb U_g\rho_{\calZ}\mathbb U_g^\dagger$ gives \eqref{eq:cq_swap_overlaps}, and the SWAP test estimates each overlap directly \cite{Buhrman2001SwapTest}.  A SWAP measurement has outcomes in $[-1,1]$, so each empirical overlap has variance at most $1/M$; the two estimates are independent, which gives the variance bound.  For the recovery statement, Hoeffding's inequality applied to the two empirical overlaps and a union bound over $\calS_0$ give $\sup_{g\in\calS_0}\abs{\widehat D_{\mathrm{CQ}}(g)-D_{\mathrm{CQ}}(g)}<\gamma_{\mathrm{CQ}}/2$ under the stated budget, and the threshold rule then separates the zero-discrepancy elements from those above the margin.
\end{proof}

Averaging these overlap differences over the group estimates asymmetry of the joint input--label mean state. The class-difference gate of \cite{CompanionRelease} instead uses signed class-pair sampling. Their exact relation is given by \cref{prop:cq_breaking_mass}. The group-independent purity estimate can be shared across candidates; $M$ pairs for that estimate and $M$ pairs for each transformed overlap use $2M(1+|\calS_0|)$ independently prepared joint-state copies.

\begin{theorem}[CQ task discrepancy]
\label{thm:cq_task_discrepancy}
\begin{enumerate}
\item[(a)] \emph{Quantum-kernel MMD identity.}
Let $\Phi(x,y):=\rho_x\otimes\ket y\!\bra y$ be viewed as a feature map into the Hilbert--Schmidt space, and define
\begin{equation}
k_{\mathrm Q}\bigl((x,y),(x',y')\bigr)
:=\Tr(\rho_x\rho_{x'})\,\abs{\ip{y}{y'}}^2.
\label{eq:quantum_task_kernel}
\end{equation}
If $Q_g:=T_g\push P$ and the action satisfies $\Phi(T_gz)=\mathbb U_g\Phi(z)\mathbb U_g^\dagger$, then
\begin{equation}
D_{\mathrm{CQ}}(g)=\frac12\MMD_{k_{\mathrm Q}}^2(P,Q_g).
\label{eq:cq_mmd_identity}
\end{equation}
If the mean embedding $P\mapsto\E_P\Phi(Z)$ is injective on the class of task distributions under consideration, then $D_{\mathrm{CQ}}(g)=0$ is equivalent to $T_g\push P=P$.  In particular, for finite input and label alphabets encoded by mutually orthogonal basis states, and for group actions represented by the corresponding permutation unitaries, $D_{\mathrm{CQ}}(g)=0$ if and only if the original joint law is invariant under $g$.
\item[(b)] \emph{Shadow estimator without global SWAP tests.}
Let $(x_i,y_i)_{i=1}^{N}$ with $N\ge2$ be independent labeled quantum examples, and let $\widehat\rho_i$ be an independent unbiased classical shadow estimator of $\rho_{x_i}$.  For $i\ne j$ define
\begin{align}
\widehat k_{ij}^{0}
&:=\Tr(\widehat\rho_i\widehat\rho_j)\,\abs{\ip{y_i}{y_j}}^2,\\
\widehat k_{ij}^{g}
&:=\Tr(\widehat\rho_iU_g\widehat\rho_jU_g^\dagger)\,
\abs{\bra{y_i}V_g\ket{y_j}}^2 .
\end{align}
Then the second-order $U$-statistic
\begin{equation}
\widehat D_{\mathrm{CQ}}^{\mathrm{sh}}(g)
:=\frac{1}{N(N-1)}
\sum_{i\ne j}\left(\widehat k_{ij}^{0}-\widehat k_{ij}^{g}\right)
\label{eq:shadow_cq_u_stat}
\end{equation}
is an unbiased estimator of $D_{\mathrm{CQ}}(g)$.  The same bank of classical-shadow snapshots can be reused for every candidate $g\in\calS_0$, so dictionary validation along the CQ task discrepancy avoids global SWAP measurements and requires no additional quantum measurements after the shadows have been collected.

Let $h_g(i,j)$ be the symmetrization of $\widehat k_{ij}^{0}-\widehat k_{ij}^{g}$, and suppose its centered value $H_g=h_g(1,2)-D_{\mathrm{CQ}}(g)$ satisfies the uniform moment bound $\E\exp(tH_g)\le\exp(t^2\nu_{\mathrm{CQ}}^2/2)$ for $|t|\le1/\nu_{\mathrm{CQ}}$ and every $g\in\calS_0$. This specifies the sub-exponential scale of the pair kernel, including both example and measurement randomness.  Let
\begin{equation}
\calS_\star^{\mathrm{CQ}}:=\{g\in\calS_0:D_{\mathrm{CQ}}(g)=0\},
\qquad
D_{\mathrm{CQ}}(g)\ge\gamma_{\mathrm{CQ}}
\quad\text{for every }g\in\calS_0\setminus\calS_\star^{\mathrm{CQ}} .
\end{equation}
There is a universal constant $C$ such that, if
\begin{equation}
N\ge C\,\left(\nu_{\mathrm{CQ}}^{2}\gamma_{\mathrm{CQ}}^{-2}+\nu_{\mathrm{CQ}}\gamma_{\mathrm{CQ}}^{-1}\right)
\log\frac{2\abs{\calS_0}}{\delta},
\label{eq:shadow_cq_sample}
\end{equation}
then the threshold rule
\begin{equation}
\widehat{\calS}_{\mathrm{sh}}^{\mathrm{CQ}}
:=
\left\{g\in\calS_0:
\widehat D_{\mathrm{CQ}}^{\mathrm{sh}}(g)\le\gamma_{\mathrm{CQ}}/2
\right\}
\end{equation}
recovers $\calS_\star^{\mathrm{CQ}}$ with probability at least $1-\delta$.
\end{enumerate}
\end{theorem}

\begin{proof}
\emph{Part~(a).}
The mean embedding of $P$ under the Hilbert--Schmidt feature map is exactly $\rho_{\calZ}$.  The mean embedding of $Q_g$ is $\mathbb U_g\rho_{\calZ}\mathbb U_g^\dagger$.  Therefore
\begin{equation}
\MMD_{k_{\mathrm Q}}^2(P,Q_g)
=\left\|\rho_{\calZ}-\mathbb U_g\rho_{\calZ}\mathbb U_g^\dagger\right\|_{\mathrm{HS}}^2
=2D_{\mathrm{CQ}}(g).
\end{equation}
The equivalence for separating encodings follows from injectivity of the mean embedding on the specified distribution class.  Orthogonal classical encodings make $\rho_{\calZ}$ diagonal with entries equal to the probabilities of $(x,y)$, so equality of CQ states is equality of the joint probability table.

\emph{Part~(b).}
For $i\ne j$, the two shadows are independent and unbiased.  Hence
\begin{equation}
\E\widehat k_{ij}^{0}
=
\E\,\Tr(\rho_{x_i}\rho_{x_j})\abs{\ip{y_i}{y_j}}^2
=\Tr(\rho_{\calZ}^{2}),
\end{equation}
and
\begin{equation}
\E\widehat k_{ij}^{g}
=
\E\,\Tr(\rho_{x_i}U_g\rho_{x_j}U_g^\dagger)
\abs{\bra{y_i}V_g\ket{y_j}}^2
=
\Tr(\rho_{\calZ}\mathbb U_g\rho_{\calZ}\mathbb U_g^\dagger).
\end{equation}
Their difference has expectation $D_{\mathrm{CQ}}(g)$ by \eqref{eq:cq_discrepancy}.  The condition $i\ne j$ removes the same-shadow self-overlap bias and makes \eqref{eq:shadow_cq_u_stat} a genuine unbiased order-two $U$-statistic.  Under the stated sub-exponential scale bound, the permutation-pairing argument used in \cref{thm:lie_recovery}, now with the assumed sub-exponential moment-generating-function bound for the symmetrized pair kernel, gives an error strictly below $\gamma_{\mathrm{CQ}}/2$ for a fixed $g$ under the sample condition \eqref{eq:shadow_cq_sample}.  A union bound over $\calS_0$ and the margin assumption then give exact recovery by the threshold rule.
\end{proof}

The two implementations have different resource costs. Paired-copy measurements estimate global overlaps with bounded outcomes. Reusable shadows replace those measurements by classical processing of a fixed snapshot bank, whose overlap scale $\nu_{\mathrm{CQ}}$ can grow exponentially with the number of qubits. A feature map formed from bounded-weight local observables or reduced states can have polynomial estimator scale when its feature count and coefficients are also controlled. Such a local feature map defines its own mean-embedding discrepancy; it must intertwine the candidate action, or be closed under its orbit, for the symmetry interpretation to apply.

For a non-injective mean embedding, zero discrepancy identifies symmetry of the encoded task mean. Distinct distributions can have the same mean even when individual inputs have distinct encodings. This is the information seen by linear expectation readouts; nonlinear functions of separately estimated features can depend on higher moments. A zero encoded discrepancy therefore supports retention only for the declared representation and predictor interface.

The observable-level and task-level tests answer different questions.  Let $X$ be a quantum input state $\rho_X$ and let a scalar label be generated by an observable $A$ through $Y=\Tr(A\rho_X)+\xi$.  If $A$ is invariant under $g$, the state distribution is $g$-invariant, and $\xi$ is independent noise with a fixed law, then the joint law is also invariant.  If $A$ is odd or otherwise nontrivial under $g$ while the output is treated as a scalar invariant label, then the same state distribution may produce a task-level symmetry mismatch.  This is the typical situation for order-parameter learning: the states may be sampled from a symmetric finite-size ensemble, while the label or diagnostic intentionally measures a symmetry-breaking direction.

An observable-invisible element is a candidate physical constraint. Task validation determines its compatible label action: an odd order parameter can respect sign equivariance while violating invariant-label symmetry. Failure of a full joint-law test does not alone prove excess prediction bias, since the regression function may remain equivariant even when label noise is asymmetric. The risk criterion is the projection residual in \cref{thm:projection_bias}. The relation between task-state asymmetry and a breaking readout is quantified by the following identity.

\begin{proposition}[Task discrepancy and breaking mass]
\label[proposition]{prop:cq_breaking_mass}
Let $Y$ take finitely many values with probabilities $p_y$, let $\rho_y$ be the class-conditional input states, and let $\rho_{\calZ}=\sum_y p_y\,\rho_y\otimes\ket y\!\bra y$ carry the joint action $\mathbb U_g=U_g\otimes V_g$ of a compact group $H$.  Write $\calT^{\mathbb U}_H$ for the Hilbert--Schmidt projection onto the commutant of $\{\mathbb U_g\}_{g\in H}$, and $\calT_H$ for the corresponding projection for $\{U_g\}_{g\in H}$.  Then
\begin{equation}
\int_H D_{\mathrm{CQ}}(g)\,dg=\bigl\|(\mathrm{id}-\calT^{\mathbb U}_H)\rho_{\calZ}\bigr\|_{\mathrm{HS}}^2 .
\label{eq:bridge_general}
\end{equation}
For a binary $\mathbb Z_2=\{e,u\}$ action, the trivial label action uses $V_u=I_2$, while the sign action $y\mapsto-y$ is implemented on the classical register by $V_u\ket{+}=\ket{-}$ and $V_u\ket{-}=\ket{+}$. For the trivial label representation the right-hand side equals $\sum_y p_y^2\|(\mathrm{id}-\calT_H)\rho_y\|_{\mathrm{HS}}^2$.  For $H=\mathbb Z_2=\{e,u\}$, balanced binary labels, $S=\rho_++\rho_-$, $\Delta=\rho_+-\rho_-$ and $\calT_H(X)=\tfrac12(X+U_uXU_u^\dagger)$,
\begin{align}
D_{\mathrm{CQ}}(u;\mathrm{triv})&=\tfrac14\bigl(\|(\mathrm{id}-\calT_H)S\|_{\mathrm{HS}}^2
 +\|(\mathrm{id}-\calT_H)\Delta\|_{\mathrm{HS}}^2\bigr),
\label{eq:bridge_z2_trivial}\\
D_{\mathrm{CQ}}(u;\mathrm{sign})&=\tfrac14\bigl(\|(\mathrm{id}-\calT_H)S\|_{\mathrm{HS}}^2
 +\|\calT_H\Delta\|_{\mathrm{HS}}^2\bigr).
\label{eq:bridge_z2_sign}
\end{align}
Finally, for binary labels, the trivial label representation, and any compact $H$ and positive class probabilities,
\begin{equation}
\|(\mathrm{id}-\calT_H)\Delta\|_{\mathrm{HS}}^2
\le 2\sum_{y}\|(\mathrm{id}-\calT_H)\rho_y\|_{\mathrm{HS}}^2
\le \frac{2}{\min_y p_y^2}\int_H D_{\mathrm{CQ}}(g)\,dg .
\label{eq:bridge_inequality}
\end{equation}
\end{proposition}

\begin{proof}
By \eqref{eq:cq_discrepancy}, $\int_H D_{\mathrm{CQ}}(g)\,dg=\Tr(\rho_{\calZ}^2)-\int_H\Tr(\rho_{\calZ}\mathbb U_g\rho_{\calZ}\mathbb U_g^\dagger)\,dg$, and the second term equals $\Tr(\rho_{\calZ}\calT^{\mathbb U}_H(\rho_{\calZ}))=\|\calT^{\mathbb U}_H\rho_{\calZ}\|_{\mathrm{HS}}^2$ because $\calT^{\mathbb U}_H$ is a self-adjoint idempotent; this gives \eqref{eq:bridge_general}.  For the trivial label representation $\calT^{\mathbb U}_H=\calT_H\otimes\mathrm{id}$ acts blockwise on the label register and the blocks are Hilbert--Schmidt orthogonal, which gives the class-weighted sum.  For $\mathbb Z_2$ and the trivial representation the label blocks of $\rho_{\calZ}-\mathbb U_u\rho_{\calZ}\mathbb U_u^\dagger$ are $(\mathrm{id}-\calT_H)\rho_\pm=\tfrac12\bigl((\mathrm{id}-\calT_H)S\pm(\mathrm{id}-\calT_H)\Delta\bigr)$, and the cross terms cancel in the sum of the two squared norms, giving \eqref{eq:bridge_z2_trivial}.  For the sign representation the blocks are $\tfrac12(\rho_+-U_u\rho_-U_u^\dagger)$ and $\tfrac12(\rho_--U_u\rho_+U_u^\dagger)$, with $\rho_+-U_u\rho_-U_u^\dagger=(\mathrm{id}-\calT_H)S+\calT_H\Delta$ and the opposite sign in the second block; the two summands lie in orthogonal ranges, so their squared norms add and \eqref{eq:bridge_z2_sign} follows.  For \eqref{eq:bridge_inequality}, $\mathrm{id}-\calT_H$ is a contraction and $\Delta=\sum_y\varepsilon_y\rho_y$ with $\varepsilon_y\in\{\pm1\}$ in the binary case, so the triangle inequality and $(a+b)^2\le 2a^2+2b^2$ give the first bound; the second follows from $\min_yp_y^2\sum_y\|(\mathrm{id}-\calT_H)\rho_y\|^2\le\sum_yp_y^2\|(\mathrm{id}-\calT_H)\rho_y\|^2$.
\end{proof}

The group average in \eqref{eq:bridge_general} is the asymmetry mass of the joint CQ state. The binary task breaking mass of \cite{CompanionRelease} is instead $\|(\mathrm{id}-\calT_H)\Delta\|_{\mathrm{HS}}^2$. Equations~\eqref{eq:bridge_z2_trivial}--\eqref{eq:bridge_inequality} relate these distinct quantities, including the label representation and class probabilities. For a general label action, the joint-state condition includes all input--label intertwining blocks and need not reduce to a single condition on $\Delta$.

For the spin-flip example of \cite{CompanionRelease}, $U_u=X^{\otimes n}$ exchanges the two class states. Then $S$ is even and $\Delta$ is odd, so $D_{\mathrm{CQ}}(u;\mathrm{sign})=0$ while $D_{\mathrm{CQ}}(u;\mathrm{triv})=\|\Delta\|_{\mathrm{HS}}^2/4$. An equivariant predictor with the sign action and an invariant predictor enlarged by an odd readout can both access this signal; they are different model constraints. Sector masses identify candidate label types, while the joint CQ test checks the full input--label action. For the trivial label action, \eqref{eq:bridge_inequality} implies that exact CQ invariance closes the binary breaking-mass gate, and supplies the conversion factor for approximate invariance. A passed sign-action test can coexist with a positive odd-sector gate. All empirical conclusions use confidence intervals at the declared tolerance.

The same averaging identifies the observable-level statistic.  If $\calO$ is an orthonormal Hermitian basis, so that the probes are informationally complete, then $D_{\calO}(g;\rho)=\tfrac1m\|\rho-U_g\rho U_g^\dagger\|_{\mathrm{HS}}^2$, and averaging over $H$ gives $\int_HD_{\calO}(g;\rho)\,dg=\tfrac2m\|(\mathrm{id}-\calT_H)\rho\|_{\mathrm{HS}}^2$.  The Hilbert--Schmidt asymmetry tests of \cite{Bandyopadhyay2023} measure this state-asymmetry mass through two-copy overlaps, while \cite{LaBordeRethinasamyWilde2024} develops complementary coherent projector circuits; \eqref{eq:Dobs} is its local-probe, per-element refinement, which is what makes element-wise dictionary recovery possible at a cost set by probe locality rather than by $\dim\calH$.

\section{Symmetry mismatch and statistical complexity}
\label{sec:mismatch}

Discovery and release are useful only if the statistical consequences of wrong symmetry choices are explicit.  There are two complementary errors.  Imposing too large a group removes task-relevant functions and creates bias.  Relaxing the group enlarges the commutant and changes the capacity term for a specified norm-constrained readout class.

\subsection{Projection bias from imposing too much symmetry}

Let $G$ act on an input space $\calX$ and an output Hilbert space $\calY$.  Let $P_X$ be $G$-invariant, and define the unitary action on square-integrable predictors by
\begin{equation}
(\mathscr U_gf)(x):=\rho_{\calY}(g)f(g^{-1}\cdot x).
\end{equation}
The $G$-equivariant function class is the fixed-point subspace
\begin{equation}
\calF_G:=\{f:\mathscr U_gf=f\text{ for all }g\in G\}.
\end{equation}
Let $\Pi_G$ be the orthogonal projection onto $\calF_G$, implemented by Haar averaging when $G$ is compact.

\begin{theorem}[Irreducible bias under an overly large symmetry]
\label{thm:projection_bias}
Let $\E\|Y\|^2<\infty$ and $f_\star(x)=\E[Y\mid X=x]$. For the squared risk $\mathcal R(f)=\E\|f(X)-Y\|^2$,
\begin{equation}
\inf_{f\in\calF_G}\mathcal R(f)-\mathcal R(f_\star)
=\|f_\star-\Pi_Gf_\star\|_{L^2(P_X)}^2.
\end{equation}
Moreover, for any $g\in G$,
\begin{equation}
\|f_\star-\Pi_Gf_\star\|_{L^2(P_X)}
\ge \frac12\|f_\star-\mathscr U_gf_\star\|_{L^2(P_X)}.
\end{equation}
\end{theorem}

\begin{proof}
The representation $\mathscr U_g$ is unitary on $L^2(P_X;\calY)$ because $P_X$ is invariant and $\rho_{\calY}$ is unitary.  Hence $\Pi_G=\int_G\mathscr U_g\,dg$ is the orthogonal projection onto the fixed-point subspace.  Orthogonal decomposition gives
\begin{equation}
f_\star=\Pi_Gf_\star+(f_\star-\Pi_Gf_\star),
\qquad
f_\star-\Pi_Gf_\star\perp\calF_G,
\end{equation}
For every square-integrable predictor, conditional expectation gives $\mathcal R(f)-\mathcal R(f_\star)=\|f-f_\star\|_{L^2(P_X)}^2$, and Pythagoras proves the risk identity.  For the second claim, use $\mathscr U_g\Pi_Gf_\star=\Pi_Gf_\star$ and $\|I-\mathscr U_g\|\le2$:
\begin{equation}
\|f_\star-\mathscr U_gf_\star\|=
\|(I-\mathscr U_g)(f_\star-\Pi_Gf_\star)\|
\le2\|f_\star-\Pi_Gf_\star\|.
\end{equation}
\end{proof}

\subsection{Commutant dimension and logarithmic generalization scale}

Let a subgroup $H\le G$ be enforced on observables through the commutation relation $[O,U_h]=0$ for all $h\in H$.  The admissible readout space is the Hermitian part of the commutant,
\begin{equation}
\calA_H:=\Comm(H)\cap\Herm(\calH),
\qquad
m_H:=\dim_{\R}\calA_H.
\end{equation}
Capacity depends on the norm used to constrain the readout. For operator-norm-bounded basis functions, an $\ell_1$ coefficient ball has a logarithmic dictionary-size bound by Massart's finite-class lemma \cite{Massart2000}. This sparsity-promoting constraint specifies the model class; it does not assert that every feasible coefficient vector has small support. A fixed Hilbert--Schmidt ball has a different geometry and can admit dimension-free bounds because density matrices have Hilbert--Schmidt norm at most one. The following statement is the standard finite-dictionary bound specialized to a symmetric readout basis \cite{BartlettMendelson2002}.

\begin{theorem}[Logarithmic commutant-dimension generalization scale]
\label{thm:commutant_bound}
Fix $H$, $W$, and a Hermitian operator basis $\{A_i\}_{i=1}^{m_H}$ of $\calA_H$ with $\norm{A_i}_{\mathrm{op}}\le 1$, all chosen independently of the risk-estimation data.  Consider symmetric observables
\begin{equation}
O_w=\sum_{i=1}^{m_H}w_iA_i,
\qquad
\norm{w}_1\le W.
\end{equation}
Let $\mathcal R(O)=\E\ell(\Tr(O\rho),Y)$ and $\widehat{\mathcal R}_N(O)=N^{-1}\sum_{s=1}^N\ell(\Tr(O\rho_s),y_s)$. Assume the loss takes values in $[0,B_\ell]$ and is $L_\ell$-Lipschitz in its scalar prediction. Then, for $N$ independent and identically distributed labeled examples, with probability at least $1-\delta$,
\begin{equation}
\sup_{\norm{w}_1\le W}
\abs{\mathcal R(O_w)-\widehat{\mathcal R}_N(O_w)}
\le
2L_\ell W\sqrt{\frac{2\log(2m_H)}{N}}
+B_\ell\sqrt{\frac{\log(2/\delta)}{2N}}.
\end{equation}
\end{theorem}

\begin{proof}
For density operators $\rho$, each coordinate function $a_i(\rho)=\Tr(A_i\rho)$ satisfies $\abs{a_i(\rho)}\le\norm{A_i}_{\mathrm{op}}\norm{\rho}_1\le1$.  The empirical Rademacher complexity of the linear class is
\begin{equation}
\mathfrak R_N
\le
W\,\E_\sigma
\max_{1\le i\le m_H}
\left|
\frac1N\sum_{s=1}^{N}\sigma_s a_i(\rho_s)
\right|.
\end{equation}
Applying Massart's lemma to the finite class $\{\pm a_i:1\le i\le m_H\}$ gives
\begin{equation}
\mathfrak R_N\le W\sqrt{\frac{2\log(2m_H)}{N}}.
\end{equation}
The Lipschitz contraction inequality multiplies this by $L_\ell$, and the standard symmetrization bound contributes the leading factor $2$.  Apply bounded differences separately to $\sup_O(\mathcal R-\widehat{\mathcal R}_N)$ and $\sup_O(\widehat{\mathcal R}_N-\mathcal R)$, each with failure budget $\delta/2$. The deterministic Rademacher bound above controls the expectation in both directions and gives the stated two-sided inequality.
\end{proof}

The bound holds conditionally for a group and basis selected from independent discovery data. Selecting among a fixed family of $J$ such classes on the same risk sample requires a simultaneous bound, obtained here by replacing $\log(2/\delta)$ with $\log(2J/\delta)$. The $\ell_2$ and soft-release bounds of \cite{CompanionRelease} use data-dependent feature covariance instead. These are distinct norm-constrained classes, so their comparison depends on the radii, basis normalization, and task approximation error.

For $n$ qubits, permutation symmetry gives an explicit change in the dictionary-size factor at fixed $W$.

\begin{proposition}[Permutation-symmetric dimension collapse]
\label[proposition]{prop:schur_dim}
On $\calH_n=(\C^2)^{\otimes n}$,
\begin{equation}
m_{\{e\}}=\dim_\R\Herm(\calH_n)=4^n,
\end{equation}
whereas for the full permutation group $S_n$ acting by permuting tensor factors,
\begin{equation}
m_{S_n}=\binom{n+3}{3}.
\end{equation}
\end{proposition}

\begin{proof}
The first identity is dimension counting.  For the second, the commutant of the permutation action is the image of the collective $\End(\C^2)$ action.  As a complex vector space it is isomorphic to $\operatorname{Sym}^n(\C^4)$, whose dimension is $\binom{n+3}{3}$ \cite{FultonHarris1991}.  Since the Hermitian part of a finite-dimensional $*$-algebra has real dimension equal to the complex dimension of the algebra, the same number is $m_{S_n}$.
\end{proof}

Excessive symmetry excludes part of the regression function and incurs the projection bias in \cref{thm:projection_bias}. Enlarging an $\ell_1$ readout dictionary changes the capacity bound through $\sqrt{\log m_H}$. At fixed coefficient radius and normalized basis observables, the trivial subgroup gives $O(W\sqrt{n/N})$, whereas full permutation symmetry gives $O(W\sqrt{\log n/N})$. These rates quantify the estimation term; whether the enlargement improves prediction depends on the approximation error it removes.

\section{Retaining and releasing symmetry}
\label{sec:release}

Once a subgroup has been discovered, the model should treat it as a retained backbone whose release remains testable.

\subsection{Hard, soft, and released models}

Let $H_{\mathrm{ret}}\le G$ be the closed subgroup retained after discovery and task validation. A hard model restricts its predictor or circuit to the corresponding invariant or equivariant class. A released model activates generators with a nonzero component outside the commutant. For a fixed radius constraint on the branch coordinates $\phi$, a soft gate varies their allowed amplitude:
\begin{equation}
U(\theta,\phi;\alpha)=U_{\mathrm{br}}\bigl(s(\alpha)\phi\bigr)U_{\mathrm{sym}}(\theta),
\qquad s(\alpha)=\frac1{1+\alpha},\qquad \alpha\in[0,\infty],
\label{eq:soft_gate}
\end{equation}
with $s(\infty)=0$. The hard model has a closed gate, and $\alpha=0$ allows the full prescribed branch radius. With unrestricted coordinates, any positive gate is only a reparameterization. A penalty can likewise control release when it depends on the current predictor or model state; a discrepancy of fixed input data is constant with respect to training parameters.

For a linear loss $\calL(\beta)=\Tr(Ce^{-i\beta B_\perp}\rho_0e^{i\beta B_\perp})$, define
\begin{equation}
\kappa_B:=\Tr\bigl([B_\perp,[B_\perp,C]]\rho_0\bigr)=-\calL''(0).
\label{eq:release_curvature}
\end{equation}
The expansion is $\calL(\beta)=\calL(0)+\beta\Tr(i[B_\perp,C]\rho_0)-\beta^2\kappa_B/2+O(\beta^3)$. Thus positive $\kappa_B$ gives second-order descent at a stationary release point. Nonlinear supervised losses require their chain-rule Hessian correction. Simultaneous measurement bounds, release amplitudes, and independent model validation determine whether this local gain warrants activation \cite{CompanionRelease}. The geometric question is how a proposed branch couples to the retained backbone.

\subsection{Geometric decoupling at the symmetry release point}
\label{subsec:release_fisher_geometry}

At an invariant pure-state backbone, the Fisher covariance pairs only compatible representation sectors. A generator with zero group average is consequently orthogonal to every retained symmetric tangent. The resulting block structure separates geometric normalization of the backbone and breaking coordinates at the release point. For normalized pure states, our quantum Fisher information matrix convention is \cite{BraunsteinCaves1994,Stokes2020QNG}
\begin{equation}
F_{ab}=4\operatorname{Re}\left(\braket{\partial_a\psi}{\partial_b\psi}-\braket{\partial_a\psi}{\psi}\braket{\psi}{\partial_b\psi}\right).
\label{eq:pure_qfim}
\end{equation}

\begin{theorem}[Fisher block-orthogonality at symmetry release]
\label{thm:fisher_release_orthogonality}
Let $H_{\mathrm{ret}}$ be a compact retained symmetry group represented by unitaries $W_h$ on the Hilbert space.  Suppose the backbone density operator $\rho_{\mathrm{sym}}(\theta)=\ket{\psi_{\mathrm{sym}}(\theta)}\bra{\psi_{\mathrm{sym}}(\theta)}$ is $H_{\mathrm{ret}}$-invariant and its local tangent generators are $H_{\mathrm{ret}}$-symmetric Hermitian operators:
\begin{equation}
W_h\rho_{\mathrm{sym}}(\theta)W_h^\dagger=\rho_{\mathrm{sym}}(\theta),\qquad
\partial_{\theta_j}\ket{\psi_{\mathrm{sym}}(\theta)}
=-iA_j(\theta)\ket{\psi_{\mathrm{sym}}(\theta)},\qquad
W_hA_j(\theta)W_h^\dagger=A_j(\theta).
\end{equation}
Let the released branch be
\begin{equation}
\ket{\psi(\theta,\beta)}
=e^{-i\beta B_\perp}\ket{\psi_{\mathrm{sym}}(\theta)},
\end{equation}
where the breaking generator has no trivial $H_{\mathrm{ret}}$-component,
\begin{equation}
\int_{H_{\mathrm{ret}}} W_hB_\perp W_h^\dagger\,dh=0 .
\label{eq:release_nontrivial_component}
\end{equation}
Then the pure-state quantum Fisher information matrix (QFIM) at $\beta=0$ is exactly block diagonal:
\begin{equation}
\FIM_{\theta_j,\beta}(\theta,0)=0
\qquad \text{for all }j,
\end{equation}
and
\begin{equation}
\FIM(\theta,0)
=
\FIM_{\theta\theta}(\theta,0)
\oplus
4\Var_{\psi_{\mathrm{sym}}(\theta)}(B_\perp).
\label{eq:release_qfim_block}
\end{equation}
Here $\Delta B=B-\langle B\rangle_\theta I$ and $\{A,B\}=AB+BA$. For a multi-parameter released block generated by $B_{\perp,a}$ satisfying \eqref{eq:release_nontrivial_component} for each $a$, the full QFIM has the block form
\begin{equation}
\FIM(\theta,0)=
\begin{pmatrix}
\FIM_{\theta\theta}(\theta,0) & 0\\
0 & \FIM_{\beta\beta}(\theta,0)
\end{pmatrix},
\qquad
(\FIM_{\beta\beta})_{ab}
=
2\,\mathrm{Re}\,
\bra{\psi_{\mathrm{sym}}(\theta)}
\{\Delta B_{\perp,a},\Delta B_{\perp,b}\}
\ket{\psi_{\mathrm{sym}}(\theta)}.
\end{equation}
If the model is simultaneously released along one-parameter breaking branches generated by \(B_{\perp}^{(1)},\ldots,B_{\perp}^{(r)}\), each generator has no trivial \(H_{\mathrm{ret}}\)-component, and each \(B_{\perp}^{(k)}\) lies in a subspace \(\mathcal V_k\subseteq\Herm(\calH)\) that is irreducible as a \emph{real} representation of \(H_{\mathrm{ret}}\) under conjugation, with \(\mathcal V_1,\ldots,\mathcal V_r\) nontrivial and pairwise non-isomorphic, then at \(\bm\beta=0\)
\begin{equation}
\FIM(\theta,\bm 0)
=
\FIM_{\theta\theta}(\theta,0)
\oplus
\bigoplus_{k=1}^{r}
4\Var_{\psi_{\mathrm{sym}}(\theta)}\!\left(B_{\perp}^{(k)}\right).
\end{equation}
The inverse or pseudoinverse therefore acts separately on the retained block and these pairwise non-isomorphic breaking blocks at the release point.
\end{theorem}

\begin{proof}
Write $\langle O\rangle_\theta=\bra{\psi_{\mathrm{sym}}(\theta)}O\ket{\psi_{\mathrm{sym}}(\theta)}$.  Since the state is invariant under $H_{\mathrm{ret}}$,
\begin{equation}
\langle O\rangle_\theta
=
\left\langle
\int_{H_{\mathrm{ret}}} W_hOW_h^\dagger\,dh
\right\rangle_\theta
\end{equation}
for every operator $O$.  With $O=B_\perp$, equation \eqref{eq:release_nontrivial_component} gives $\langle B_\perp\rangle_\theta=0$.  Since $A_j(\theta)$ is $H_{\mathrm{ret}}$-symmetric, the operators $A_j(\theta)B_\perp$ and $B_\perp A_j(\theta)$ have the same nontrivial $H_{\mathrm{ret}}$-type as $B_\perp$, so their Haar averages vanish.  Therefore
\begin{equation}
\langle A_j(\theta)B_\perp\rangle_\theta
=
\langle B_\perp A_j(\theta)\rangle_\theta
=0.
\end{equation}
At $\beta=0$, the pure-state QFIM cross entry is the symmetric connected covariance
\begin{equation}
\FIM_{\theta_j,\beta}(\theta,0)
=
2\,\mathrm{Re}\,
\left\langle
\Delta A_j(\theta)\Delta B_\perp
+\Delta B_\perp\Delta A_j(\theta)
\right\rangle_\theta,
\end{equation}
which is zero.  The diagonal release entry follows by applying the same pure-state covariance formula to $B_\perp$.  The multi-parameter retained--release statement is identical for each cross block entry. For distinct one-parameter branches \(k\ne \ell\), the cross Fisher entry is the symmetric connected covariance of \(B_{\perp}^{(k)}\) and \(B_{\perp}^{(\ell)}\).  Since the backbone state is invariant, this covariance equals its Haar average over \(H_{\mathrm{ret}}\).  Conjugation acts orthogonally on $\Herm(\calH)$ for the Hilbert--Schmidt inner product, so each $\mathcal V_k$ is self-dual as a real representation, and an invariant bilinear pairing between $\mathcal V_k$ and $\mathcal V_\ell$ is an element of $\mathrm{Hom}_{H_{\mathrm{ret}}}(\mathcal V_k,\mathcal V_\ell^*)\cong\mathrm{Hom}_{H_{\mathrm{ret}}}(\mathcal V_k,\mathcal V_\ell)$, which vanishes for non-isomorphic real irreducibles.  The averaged cross covariance is therefore zero.  Reality is what makes non-isomorphism sufficient: counting complex types instead gives the wrong answer, since the charge-$\pm q$ representations of $U(1)$ are non-isomorphic yet mutually dual, and the Hermitian operators they span form a single two-dimensional real irreducible whose directions must be treated within one real block rather than separated by their complex charge labels.  The branch-diagonal entries are the pure-state variances.
\end{proof}

The Moore--Penrose inverse respects the retained--breaking Fisher decomposition at the invariant point, which separates geometric normalization of these blocks. The task Hessian can still couple directions within an isotypic block, and its statistical certificate is supplied independently by \cite{CompanionRelease}. The result complements the sector geometry of \cite{ChenQNGQuotient2026}: here the generators cross the commutant boundary. For a single fixed left-applied branch $e^{-i\beta B_\perp}$ the retained--branch Fisher cross block stays zero for all $\beta$, because both tangents are transported by the same unitary and $B_\perp$ commutes with its own flow. General multibranch parameterizations need only decouple at the origin; their subsequent coupling is handled by the Schur-complement rule of \cite{ChenQNGQuotient2026}.

\section{Optimization and finite-size diagnostics}
\label{sec:diagnostics}

After discovery, the remaining question is whether an observed symmetry-breaking signal is physically meaningful, statistically stable, or an optimization artifact.

State-preserving redundant coordinates carry no physical variation. The quantum natural gradient uses the Fisher metric to distinguish their role from weak but physical directions \cite{Amari1998Natural,BraunsteinCaves1994,Stokes2020QNG}. If a regular quotient has vertical tangent space exactly equal to $\ker\FIM$, the Moore--Penrose update
\begin{equation}
\Delta\theta=-\eta\FIM(\theta)^+\nabla\calL(\theta)
\label{eq:qng_update}
\end{equation}
is the minimum-Euclidean-norm horizontal lift of the quotient natural gradient \cite{ChenQNGQuotient2026}. Projection onto a known exact horizontal space keeps each update free of vertical displacement, including with finite-shot gradient estimates. An estimated projector generally leaves leakage. Recovering a local kernel from noisy spectra requires independently supplied gap or structural information, and does not by itself identify a global redundancy group. Nonzero Fisher eigenvalues also depend on coordinate normalization, so comparisons across models require a fixed generator normalization or reference Gram matrix.

\subsection{Observable, Fisher, and sector diagnostics}

For a control parameter $\mu$ and system size $L$, define a diagnostic triplet
\begin{equation}
\mathcal D_L(\mu):=\bigl(D_{\mathrm{sym},L}(\mu),\lambda_{\min,L}^+(\mu),p_{\mathrm{sec},L}(\mu)\bigr).
\label{eq:triplet}
\end{equation}
The first component is a symmetry-deviation statistic, for example an average of $D_{\calO}(g;\rho)$ over recovered generators or over deliberately released generators.  The second is the smallest nonzero eigenvalue of the parameter-space QFIM used by the optimizer.  The third is a representation-sector profile
\begin{equation}
p_{\lambda,L}(\mu):=\Tr(\Pi_{\lambda,L}\rho_L(\mu)),
\qquad
p_{\mathrm{sec},L}(\mu):=\{p_{\lambda,L}(\mu)\}_{\lambda\in\widehat H},
\label{eq:sector_profile}
\end{equation}
where $H$ is the chosen retained or ambient group, $\widehat H$ indexes its irreducible representations, and $\Pi_{\lambda,L}$ projects onto the corresponding state-space isotypic sector. These projectors act on $\calH$, in contrast to the operator-space representation sectors used for breaking generators.

The three components provide complementary diagnostics. Persistent probe deviations witness observable symmetry breaking; state-sector weights track how a perturbation redistributes representation content; the Fisher spectrum measures local distinguishability in the chosen coordinates. Consistency across repeated runs and system sizes strengthens their physical interpretation, but no universal release criterion requires all three to change. For example, a symmetric mixed state can occupy several sectors, and a sector population can change without creating inter-sector coherence. Raw coordinate motion along a known state-null direction is parameter drift.

At a fixed Fisher matrix, finite-shot gradient noise with covariance $\Sigma$ produces
\begin{equation}
\E\|\Delta\theta-\E\Delta\theta\|_2^2
=\eta^2\Tr(\FIM^+\Sigma\FIM^+)
\le\frac{\eta^2\|\Sigma\|_{\mathrm{op}}\rank\FIM}{(\lambda_{\min}^+)^2}.
\label{eq:qng_noise_bound}
\end{equation}
This upper bound quantifies possible noise amplification, rather than a compulsory fluctuation size \cite{ChenQNGQuotient2026}. Comparisons therefore hold the measurement budget and coordinate normalization fixed.

\subsection{Physical quantum Fisher scaling}

Physical unitary-family Fisher information has a finite-size scaling theory related to quantum geometric tensors \cite{CamposVenutiZanardi2007,Gu2010FidelityReview}.  Let
\begin{equation}
\widehat O=\sum_i \widehat o_i
\end{equation}
be a physical operator associated with a breaking direction.  For a pure state,
\begin{equation}
F^{\mathrm{phys}}_{OO}=4\Var(\widehat O)
=4\sum_{i,j}\langle \widehat o_i\widehat o_j\rangle_c.
\end{equation}
Here $\langle AB\rangle_c=\operatorname{Re}\langle AB\rangle-\langle A\rangle\langle B\rangle$ denotes the symmetrized connected correlation for Hermitian operators.

\begin{proposition}[Scaling bridge for a physical breaking operator]
\label[proposition]{prop:qfi_scaling}
Consider translation-invariant pure states on periodic cubic lattices of spatial dimension $d_{\mathrm{lat}}$ and $L$ sites per axis, with bounded local Hermitian observables and on-site variance bounded above and below by positive constants. Let $z$ and $\eta$ denote the dynamical and anomalous correlation exponents. Assume nonnegative connected correlations, with the following two-sided power-law comparison uniformly from microscopic to the largest finite-size distances:
\begin{equation}
\langle \widehat o_i\widehat o_j\rangle_c\asymp\|i-j\|^{-(d_{\mathrm{lat}}+z-2+\eta)},
\end{equation}
Then at linear size $L$,
\begin{equation}
F^{\mathrm{phys}}_{OO}\asymp
\begin{cases}
L^{d_{\mathrm{lat}}+2-z-\eta}, & z+\eta<2,\\[2mm]
L^{d_{\mathrm{lat}}}\log L, & z+\eta=2,\\[2mm]
L^{d_{\mathrm{lat}}}, & z+\eta>2,
\end{cases}
\label{eq:qfi_scaling_cases}
\end{equation}
up to nonuniversal constants and subleading corrections.
\end{proposition}

\begin{proof}
Translation invariance gives
\begin{equation}
\sum_{i,j}\langle \widehat o_i\widehat o_j\rangle_c
\asymp L^{d_{\mathrm{lat}}}\sum_r\langle \widehat o_0\widehat o_r\rangle_c.
\end{equation}
Approximating the remaining sum by an integral yields
\begin{equation}
L^{d_{\mathrm{lat}}}\int_1^L r^{d_{\mathrm{lat}}-1}r^{-(d_{\mathrm{lat}}+z-2+\eta)}\,\mathrm{d}r
=L^{d_{\mathrm{lat}}}\int_1^L r^{1-z-\eta}\,\mathrm{d}r.
\end{equation}
The integral grows as $L^{2-z-\eta}$ when $z+\eta<2$, grows logarithmically when $z+\eta=2$, and is bounded when $z+\eta>2$, giving \eqref{eq:qfi_scaling_cases}.
\end{proof}

The exponent in \cref{prop:qfi_scaling} belongs to the physical unitary family generated by $\widehat O$. Relating it to a circuit Fisher eigenvalue requires a controlled tangent map and a specified coordinate normalization. A single aligned direction supplies a directional Fisher value, not necessarily the smallest eigenvalue of a multi-parameter matrix.

\subsection{Sector drift from symmetry-breaking perturbations}

\begin{theorem}[Local sector drift from a breaking generator]
\label{thm:sector_drift}
Let $\ket{\psi_{\mathrm{sym}}}$ lie in the trivial sector, and let $\widehat O=\widehat O^\dagger$ satisfy
\begin{equation}
\Pi_q\widehat O\ket{\psi_{\mathrm{sym}}}=\widehat O\ket{\psi_{\mathrm{sym}}},
\qquad
\Pi_\lambda\widehat O\ket{\psi_{\mathrm{sym}}}=0\quad(\lambda\ne q)
\end{equation}
for one nontrivial sector $q$.  For
\begin{equation}
\ket{\psi(\phi)}=e^{-i\phi\widehat O}\ket{\psi_{\mathrm{sym}}},
\qquad
p_q(\phi)=\bra{\psi(\phi)}\Pi_q\ket{\psi(\phi)},
\end{equation}
the sector orthogonality implies $\langle\widehat O\rangle_{\psi_{\mathrm{sym}}}=0$, and
\begin{equation}
p_q(\phi)=\frac{\phi^2}{4}F^{\mathrm{phys}}_{OO}(0)+O(\phi^4).
\end{equation}
The fourth-order remainder holds at fixed system size under the stated single-sector support condition.
\end{theorem}

\begin{proof}
Put $v=\widehat O\ket{\psi_{\mathrm{sym}}}$ and $w=\widehat O^2\ket{\psi_{\mathrm{sym}}}$. Since $\Pi_qv=v$ and $\Pi_q\ket{\psi_{\mathrm{sym}}}=0$,
\begin{equation}
\Pi_q\ket{\psi(\phi)}=-i\phi v-\frac{\phi^2}{2}\Pi_qw+O(\phi^3).
\end{equation}
The quadratic contribution is $\phi^2\|v\|^2$. Its cubic cross term vanishes because $\ip{v}{\Pi_qw}=\ip{v}{w}=\langle\widehat O^3\rangle$ is real. Thus $p_q(\phi)=\phi^2\langle\widehat O^2\rangle+O(\phi^4)$. Orthogonality to the initial sector gives $\langle\widehat O\rangle=0$, and the pure-state identity $F^{\mathrm{phys}}_{OO}=4\Var(\widehat O)$ proves the claim.

\end{proof}

\subsection{Shadow estimation of physical Fisher signals}

\begin{proposition}[Shadow estimation of local order-parameter Fisher signals]
\label[proposition]{prop:shadow_fisher}
Let
\begin{equation}
\widehat O_L=\sum_{i=1}^{m_L}P_i,
\qquad m_L\le C_m L^{d_{\mathrm{lat}}},
\end{equation}
where each $P_i$ is a Pauli string of support at most $w$ and operator norm one. On a pure input state, define the physical unitary-family Fisher information
\begin{equation}
F_L:=4\Var(\widehat O_L)
=4\left(\sum_{i,j}\left\langle \frac{P_iP_j+P_jP_i}{2}\right\rangle-\left(\sum_i\langle P_i\rangle\right)^2\right).
\end{equation}
For any $0<\delta<1$ and $0<\varepsilon\le1$, local Pauli classical shadows estimate all expectations $\{\langle P_i\rangle,\langle (P_iP_j+P_jP_i)/2\rangle\}$ to additive error at most $\varepsilon$ using
\begin{equation}
N=O\!\left(3^{2w}\varepsilon^{-2}\log\frac{2m_L^2}{\delta}\right)
\end{equation}
measurements.  The resulting estimator $\widehat F_L$ satisfies, with probability at least $1-\delta$,
\begin{equation}
|\widehat F_L-F_L|\le 16m_L^2\varepsilon.
\end{equation}
If $F_L\ge cL^\gamma$ and $0<\varrho\le1$, choosing $\varepsilon=\min\{1,(c\varrho/16C_m^2)L^{\gamma-2d_{\mathrm{lat}}}\}$ gives relative error at most $\varrho$ with
\begin{equation}
N=O\!\left(3^{2w}L^{(4d_{\mathrm{lat}}-2\gamma)_+}\varrho^{-2}\log\frac{2L}{\delta}\right),
\end{equation}
where the positive part only records that the accuracy requirement stops tightening once $\varepsilon$ is order one.  For the one-dimensional transverse-field Ising critical magnetization, where $d_{\mathrm{lat}}=1$, $w=1$, and $\gamma=7/4$, this becomes $N=O(L^{1/2}\varrho^{-2}\log(L/\delta))$.
\end{proposition}

\begin{proof}
The simultaneous-prediction theorem for local Pauli shadows estimates each Pauli observable of weight at most $2w$ to additive error $\varepsilon$ with the stated sample complexity \cite{Huang2020ClassicalShadows}.  The symmetrized product $(P_iP_j+P_jP_i)/2$ is either zero or a signed Pauli observable of weight at most $2w$.  If $a_i=\langle P_i\rangle$ and $b_{ij}=\langle (P_iP_j+P_jP_i)/2\rangle$, then $|\widehat a_i-a_i|\le\varepsilon$ and $|\widehat b_{ij}-b_{ij}|\le\varepsilon$ imply
\begin{equation}
\left|\sum_{i,j}(\widehat b_{ij}-b_{ij})\right|\le m_L^2\varepsilon,
\qquad
\left|\sum_i\widehat a_i-\sum_i a_i\right|\le m_L\varepsilon.
\end{equation}
Since $|a_i|\le1$ and $\varepsilon\le1$, the square term contributes at most $3m_L^2\varepsilon$.  Multiplying by the outer factor four gives $16m_L^2\varepsilon$.  Substituting $m_L\le C_mL^{d_{\mathrm{lat}}}$ and the chosen $\varepsilon$ gives the relative-error and measurement-scaling claims.
\end{proof}

\subsection{Transverse-field Ising scaling baseline}
\label{subsec:tfim-baseline}

As a concrete finite-size baseline, consider the one-dimensional transverse-field Ising model with periodic boundary conditions,
\begin{equation}
H(h)=-\sum_{i=1}^{L}Z_iZ_{i+1}-h\sum_{i=1}^{L}X_i,
\end{equation}
which has a $\mathbb Z_2$ spin-flip symmetry and a critical point at $h_c=1$ \cite{Pfeuty1970,Sachdev2011}.  For the symmetry-breaking magnetization $\widehat M=\sum_iZ_i$, the physical Fisher signal is
\begin{equation}
F^{\mathrm{phys}}_{MM}(h,L)=4\Var_{\psi_0(h)}(\widehat M)
\end{equation}
on the symmetric finite-size ground state.  Figure~\ref{fig:tfim_ed} shows an exact-diagonalization baseline: the critical and ordered regimes grow faster with $L$ than the disordered regime.  These exponents characterize the specified magnetization-generated family and provide a physical benchmark for aligned circuit directions.

\begin{figure}[H]
\centering
\begin{tikzpicture}
\begin{loglogaxis}[
    width=0.74\textwidth,
    height=0.48\textwidth,
    xlabel={$L$},
    ylabel={$F^{\mathrm{phys}}_{MM}(h,L)$},
    xmin=7.6, xmax=14.6,
    ymin=70, ymax=900,
    xtick={8,10,12,14},
    xticklabels={8,10,12,14},
    ytick={80,100,200,400,800},
    yticklabels={80,100,200,400,800},
    grid=both,
    legend style={draw=none, fill=none, font=\small, at={(0.03,0.97)}, anchor=north west},
    legend cell align=left]
\addplot[only marks, mark=*, mark size=2.3pt] coordinates {(8,80.707) (10,102.622) (12,123.996) (14,145.070)};
\addlegendentry{$h=1.5$}
\addplot[only marks, mark=square*, mark size=2.2pt] coordinates {(8,158.086) (10,234.124) (12,322.622) (14,423.021)};
\addlegendentry{$h=1.0$}
\addplot[only marks, mark=triangle*, mark size=2.6pt] coordinates {(8,240.524) (10,375.266) (12,539.722) (14,733.927)};
\addlegendentry{$h=0.5$}
\addplot[domain=8:14, samples=80, dashed] {9.155463043361758*x^1.0479644597888862};
\addplot[domain=8:14, samples=80, dashed] {4.078660751341114*x^1.758871767320311};
\addplot[domain=8:14, samples=80, dashed] {3.8095645542821392*x^1.9934649575278782};
\end{loglogaxis}
\end{tikzpicture}
\caption{Exact-diagonalization baseline for the transverse-field Ising chain.  The dots show $F^{\mathrm{phys}}_{MM}=4\Var(\widehat M)$ for $L\in\{8,10,12,14\}$, and the dashed lines are power-law visual fits.  The fitted exponents are approximately $1.048$ for $h=1.5$, $1.759$ for $h=1.0$, and $1.993$ for $h=0.5$, matching the expected disordered, critical, and ordered trends.}
\label{fig:tfim_ed}
\end{figure}
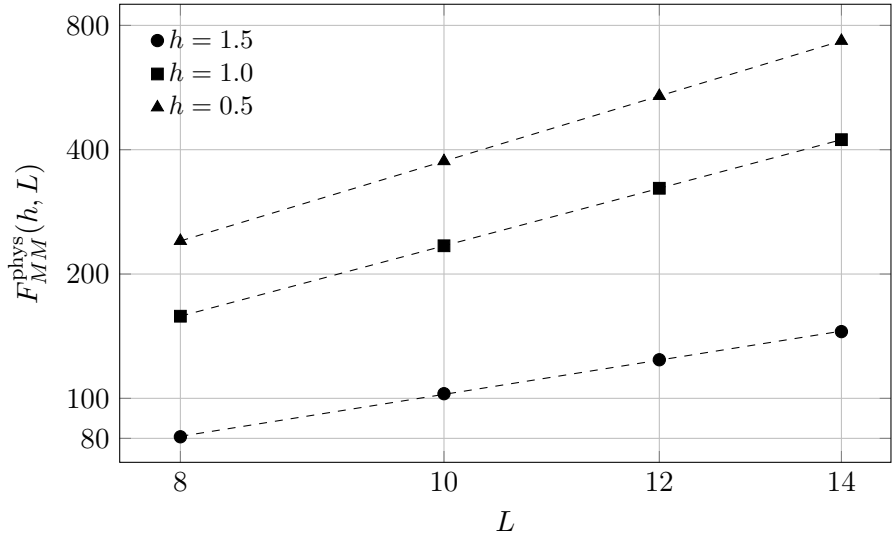

\section{Worked example: the transverse-field Ising chain}
\label{sec:worked_example}

The periodic transverse-field Ising chain permits a joint evaluation of discovery, label validation, and release geometry:
\begin{equation}
H(h)=-\sum_{i=1}^{L}Z_iZ_{i+1}-h\sum_{i=1}^{L}X_i,
\qquad L=8,
\label{eq:tfim_worked}
\end{equation}
whose $\mathbb Z_2$ spin-flip symmetry is generated by $P=\prod_iX_i$.  Two data sources are used.  The unlabeled source is the symmetric ground state $\ket{\psi_0(h)}$.  The labeled source is a pair of pinned states $\ket{\psi_\pm}=\mathrm{GS}\bigl(H(h)\mp\varepsilon M_z\bigr)$ with $M_z=\sum_iZ_i$ and $\varepsilon=0.05$, carrying labels $y=\pm1$ with equal prior.  The candidate dictionary is
\begin{equation}
\calS_0=\{P,\;T,\;R,\;X_0,\;Z_0\},
\end{equation}
where $T$ is the one-site translation, $R$ the site reflection, and $X_0,Z_0$ single-site flips that are visible on the states considered below.  The probe family is $\calO=\{Z_i\}\cup\{Z_iZ_{i+1}\}\cup\{X_i\}$, so $m=3L=24$ Pauli strings of weight at most two, and every dictionary element preserves this probe span as well as Pauli weight. Thus both the subgroup condition of \cref{thm:probe_stabilizer} and the measurement-locality condition of \cref{prop:shadow_constants} hold.

\subsection{Observable level}

On the symmetric ground state at $h=1$ the discrepancy \eqref{eq:Dobs} vanishes on $\{P,T,R\}$ and is bounded away from zero elsewhere,
\begin{equation}
D_{\calO}(X_0)=0.1368,\qquad D_{\calO}(Z_0)=0.0684,\qquad \gamma_{\calO}=0.0684 .
\end{equation}
On the pinned state the spin flip becomes visible, $D_{\calO}(P)=0.7954$, which is exactly $\tfrac43\bar m^2$ with $\bar m=\langle Z_0\rangle=0.7723$, because only the eight $Z_i$ probes out of $24$ change sign under $P$; the invisible set contracts to $\{T,R\}$ and the margin becomes $\gamma_{\calO}=0.0532$.  The symmetric ground-state probe values at $h=1$ satisfy the expected duality relation, $\langle Z_0Z_1\rangle=\langle X_0\rangle=0.6407$.

\Cref{fig:shadow_recovery} is the finite-measurement counterpart.  Pauli shadows are drawn from each state and the invisible set is estimated by thresholding at $\gamma_{\calO}/2$, once with the second-order $U$-statistic of \cref{thm:shadow_recovery} and once with the plug-in square of the empirical mean.  The unbiased estimator reaches high exact-recovery probability at a few hundred snapshots, while the plug-in estimator needs a larger budget in this example because its positive bias $\Var/M$ pushes symmetric elements above the threshold; at $M=300$ on the symmetric ground state the success probabilities are $0.92$ and $0.658$ respectively.  The gap closes as $M$ grows.

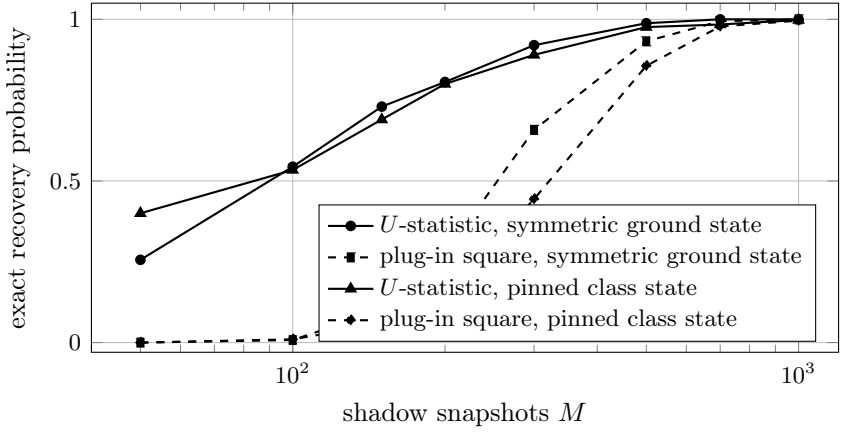
\begin{figure}[t]
\centering
\begin{tikzpicture}
\begin{semilogxaxis}[
  width=0.72\textwidth, height=6.2cm,
  xlabel={shadow snapshots $M$}, ylabel={exact recovery probability},
  xmin=40, xmax=1200, ymin=-0.03, ymax=1.05,
  legend pos=south east, legend cell align=left, legend style={font=\footnotesize},
  grid=major, tick label style={font=\footnotesize}, label style={font=\small}]
\addplot[thick, mark=*, mark size=1.6pt] coordinates {(50,0.256)(100,0.544)(150,0.730)(200,0.806)(300,0.920)(500,0.988)(700,1.000)(1000,1.000)};
\addlegendentry{$U$-statistic, symmetric ground state}
\addplot[thick, dashed, mark=square*, mark size=1.5pt] coordinates {(50,0.000)(100,0.008)(150,0.100)(200,0.274)(300,0.658)(500,0.932)(700,0.994)(1000,1.000)};
\addlegendentry{plug-in square, symmetric ground state}
\addplot[thick, mark=triangle*, mark size=2.0pt] coordinates {(50,0.400)(100,0.534)(150,0.690)(200,0.800)(300,0.890)(500,0.976)(700,0.984)(1000,0.998)};
\addlegendentry{$U$-statistic, pinned class state}
\addplot[thick, dashed, mark=diamond*, mark size=2.0pt] coordinates {(50,0.000)(100,0.010)(150,0.058)(200,0.146)(300,0.444)(500,0.856)(700,0.978)(1000,0.996)};
\addlegendentry{plug-in square, pinned class state}
\end{semilogxaxis}
\end{tikzpicture}
\caption{Observable-level dictionary recovery from Pauli classical shadows on the $L=8$ chain at $h=1$, over $500$ independent repetitions per point.  The estimated invisible set is declared by thresholding at $\gamma_{\calO}/2$; a repetition counts as a success only if the recovered set equals the population invisible set exactly.  The population margin is supplied to both threshold rules. The unbiased $U$-statistic performs better on this tested grid; the plug-in square has positive finite-sample bias. These curves compare the estimators at the stated margins and do not establish a different asymptotic recovery rate for the plug-in rule.}
\label{fig:shadow_recovery}
\end{figure}

\subsection{Task level}

The labeled source gives the opposite verdict about the same group element, and the two verdicts differ only through the label representation.  Because $P\rho_+P=\rho_-$ holds exactly, the class average $S=\rho_++\rho_-$ is $\mathbb Z_2$-invariant and the label difference $\Delta=\rho_+-\rho_-$ is odd, so $\|(\mathrm{id}-\calT_H)S\|_{\mathrm{HS}}=0$ and $\calT_H\Delta=0$.  \Cref{prop:cq_breaking_mass} then predicts
\begin{equation}
D_{\mathrm{CQ}}(P;\mathrm{triv})=\tfrac14\|\Delta\|_{\mathrm{HS}}^2
=\tfrac12\bigl(1-\abs{\braket{\psi_+}{\psi_-}}^2\bigr),
\qquad
D_{\mathrm{CQ}}(P;\mathrm{sign})=0 ,
\end{equation}
and exact diagonalization confirms it: $D_{\mathrm{CQ}}(P;\mathrm{triv})=0.5000$, $0.4595$ and $0.0833$ at $h=0.5$, $1.0$ and $1.5$.  Under an invariant scalar label the task therefore rejects the spin flip at every field, most strongly in the ordered phase where the two pinned ground states are nearly orthogonal; under a sign label the task is exactly $\mathbb Z_2$-equivariant.  This is the value of separating the levels: one data set, one group element, and two opposite retention decisions selected entirely by the output representation.  A model that imposes hard invariance here would project away the order parameter it was built to learn, at the irreducible cost of \cref{thm:projection_bias}.

\subsection{Release point}
\label{subsec:worked_release_point}

Take the backbone $\ket{\psi_{\mathrm{sym}}(\theta)}=e^{-i\theta_1M_x}e^{-i\theta_2V}\ket{{+}x}^{\otimes L}$ with $V=\sum_iZ_iZ_{i+1}$, which is $\mathbb Z_2$-invariant with invariant tangent generators, and release along $B_\perp=M_y=\sum_iY_i$, which is odd under $P$ and therefore satisfies \eqref{eq:release_nontrivial_component}.  At $(\theta_1,\theta_2)=(0.37,0.21)$ and $\beta=0$ the quantum Fisher matrix in the coordinates $(\theta_1,\theta_2,\beta)$ is
\begin{equation}
\FIM=\diag(18.628,\;32.000,\;13.108),
\end{equation}
with all cross entries below $10^{-14}$, and the branch block equals $4\Var(M_y)=13.108$ as \eqref{eq:release_qfim_block} requires.  The retained--breaking cross entries vanish at every $\theta$ by parity. The vanishing cross entry within the retained block at this reported point is a property of the chosen backbone, rather than a consequence of the retained--breaking theorem.

\subsection{Diagnosis}

Moving off the release point drives weight into the odd parity sector.  At the product point $\theta=0$, where the backbone is the bare $\ket{{+}x}^{\otimes L}$, the sector profile is closed-form, $p_{\mathrm{odd}}(\beta)=\tfrac12\bigl(1-\cos^L2\beta\bigr)$, and \cref{thm:sector_drift} predicts the initial rate $\beta^2F^{\mathrm{phys}}/4$ with $F^{\mathrm{phys}}=4\Var_{\ket{{+}x}^{\otimes L}}(M_y)=4L=32$.  Exact evaluation gives $p_{\mathrm{odd}}=0.00319$, $0.01964$ and $0.07439$ at $\beta=0.02$, $0.05$ and $0.1$, against the quadratic prediction $0.0032$, $0.0200$ and $0.0800$: the leading behavior is confirmed and the deviation grows with $\beta$ as the quartic term switches on.  Away from $\theta=0$ the same rate holds with the local Fisher value replacing $4L$; at the $(0.37,0.21)$ point of \cref{subsec:worked_release_point} the physical Fisher value is $13.108$, and the drift is smaller in proportion.  The sector component of the diagnostic is therefore quantitative, and the branch Fisher entry is its physical initial-drift coefficient. The magnetization scaling in \cref{fig:tfim_ed} concerns a separate ground-state family and supplies a finite-size physical benchmark; it does not determine the smallest eigenvalue of this variational backbone.

\subsection{Loss curvature and the release decision}

The same symmetry can support different local energy-curvature verdicts on different backbones. Take energy as the loss observable, so that $\rho_0$ in \eqref{eq:release_curvature} is the backbone state.  At the product point $\ket{{+}x}^{\otimes L}$ the release curvature along $M_y$ is positive and closed-form, $\kappa_B=4L(2-h)$, equal to $48.0$, $32.0$ and $16.0$ at $h=0.5$, $1.0$ and $1.5$ and confirmed by exact diagonalization to four digits, whereas on the exact symmetric ground state it is negative at every field, from $-57.4$ at $h=0.5$ to $-10.2$ at $h=1.5$.  Both signs are correct and they answer different questions: a mean-field product state is unstable against the odd direction, while the true ground state is an energy minimum with no descent direction at all.  The central-difference probe of \cite{CompanionRelease}, evaluated here with exact scalar loss expectations, gives $\widehat\kappa_B(\delta)=29.08$, $31.26$, $31.81$ and $31.95$ at $\delta=0.2$, $0.1$, $0.05$ and $0.025$, converging to $32$ at the expected $O(\delta^2)$ rate.

These curvature calculations identify the local task signal. A finite-shot activation decision must additionally control estimation error, the step remainder, and the statistical cost of the enlarged predictor \cite{CompanionRelease}. Discovery supplies the retained group and label action; \cref{thm:fisher_release_orthogonality} supplies their release geometry, and quotient natural gradient addresses independently identified parameter redundancy during training \cite{ChenQNGQuotient2026}.

\section{Permutation and lattice symmetries}
\label{sec:workflow}

For permutation symmetry, adjacent transpositions form a natural candidate dictionary, $\calS_0=\{(i,i+1):1\le i<n\}$. Complete fixed-weight Pauli probe spans are invariant under these swaps. Consequently, exact recovered generators define a projected-state stabilizer subgroup by \cref{thm:probe_stabilizer}. Task validation tests the same swaps together with their label action, and Schur--Weyl duality gives the commutant dimension in \cref{prop:schur_dim}. Local probe invisibility still has to be distinguished from full many-body invariance unless the candidate state variations are separated.

On a finite periodic lattice, the spatial symmetry group often has the form $G_{\mathrm{lat}}=T\rtimes P_{\mathrm{pt}}$, where $T$ is the finite translation group and $P_{\mathrm{pt}}$ is the point group. Primitive translations and point-group generators give a compact dictionary. Site permutations preserve complete fixed-weight probe spans, while more selective geometric probes require a check of their orbit span. Density-wave, nematic, and magnetization observables specify different possible label representations and breaking directions. In quantum convolutional models, equivariant pooling must intertwine the fine- and coarse-lattice actions so that architectural changes have a controlled sector interpretation \cite{CongChoiLukin2019QCNN,Nguyen2024}.

Across these settings, the retained constraint is the closed subgroup generated by transformations supported at both the declared observable and task levels. A joint confidence statement follows by allocating the failure probability across the simultaneous discovery and task events. Approximate tests provide tolerance-dependent statements, with the word bound in \eqref{eq:probe_word_bound} controlling only finite compositions. For continuous candidates, the spectral estimate gives a nearby infinitesimal subspace; noisy basis vectors need not themselves close as an exact Lie algebra. Imposing an exact fitted group therefore requires structural constraints or additional validation of its finite action.

\section{Conclusion}
\label{sec:conclusion}

The symmetry that finite measurements can identify is the stabilizer of the information those measurements retain. An invariant probe span makes this statement exact: projection converts observable invisibility into a closed-subgroup condition and its infinitesimal version into a Lie algebra. The squared-discrepancy structure then gives inverse-gap finite-dictionary recovery, with a matching qubit lower bound for the gap dependence. Together, the structural and statistical results explain when noisy generator tests support an imposed group and how approximate acceptance propagates to finite compositions.

The label action determines what that group means for learning. Characteristic kernels test the full joint law, while the classical--quantum discrepancy tests its encoded mean and connects it exactly to joint-state asymmetry. The binary identities distinguish this quantity from a breaking readout signal and explain the opposite invariant- and sign-label verdicts in the Ising example. Projection bias, norm-constrained readout capacity, and retained--breaking Fisher orthogonality describe the resulting model choice at the function, statistical, and geometric levels. They supply the inputs to statistical release and quotient-space training \cite{CompanionRelease,ChenQNGQuotient2026} without identifying physical symmetry with parameter redundancy.

Further progress requires relating measurement resolution to task loss and fitted group structure. A loss-specific comparison with probe discrepancy would turn statistically valid tolerances into task-calibrated retention thresholds. For non-invariant probe spans, smaller separating families could reduce the cost of validating finite group actions without collecting their full operator orbit. Continuous recovery also leaves a constrained estimation problem: converting a noisy tangent subspace into a group with a certified finite action. These questions concern the information and geometry needed to impose a constraint, beyond the resolution of its individual measurement tests.

\end{document}